\documentclass[11pt,a4paper]{article}

\usepackage[margin=1in]{geometry}
\usepackage{amsmath,amssymb,amsthm}
\usepackage{booktabs}
\usepackage{graphicx}
\graphicspath{{figures/}}
\usepackage{hyperref}
\usepackage[authoryear,round]{natbib} 
\usepackage[protrusion=true,expansion=false]{microtype}
\usepackage{caption}
\usepackage{subcaption}
\usepackage{enumerate}
\usepackage{float}
\newfloat{algorithm}{t}{loa}
\floatname{algorithm}{Algorithm}

\hypersetup{colorlinks=true,linkcolor=blue,citecolor=blue,urlcolor=blue}

\newtheorem{theorem}{Theorem}
\newtheorem{corollary}{Corollary}
\newtheorem{proposition}{Proposition}
\newtheorem{lemma}{Lemma}
\theoremstyle{remark}
\newtheorem{remark}{Remark}

\newcommand{\op}{\mathrm{op}}
\newcommand{\AR}{\mathrm{AR}}
\newcommand{\gaptwelve}{\mathrm{gap}_{12}}
\newcommand{\R}{\mathbb{R}}
\newcommand{\E}{\mathbb{E}}

\title{Error Propagation in Spectral Functionals of Shrinkage Covariance
Estimators:\\
Perturbation Bounds and Calibrated Inference}

\author{Ahmad Koman\\
  Department of Mathematics\\
  Chalmers University of Technology\\
  Chalmers Tv\"argata 3, 412 58 G\"oteborg, Sweden\\
  \texttt{ahmadko@chalmers.se}}

\date{\today}

\begin{document}
\maketitle

\begin{abstract}
Rolling covariance estimates feed two objects that are routinely treated as
market structure. The first is the dominant eigenspace, monitored through the projector
movement $\widehat D_{K,t}=\|\widehat P_{K,t}-\widehat P_{K,t-1}\|_F$; the second
comprises scalar spectral functionals such as the absorption ratio and the
leading-eigenvalue share.
Both fluctuate under estimation noise, and shrinkage changes the law of that
noise, so reading their movements as structural change requires calibration.
For the eigenspace, we derive a first-order null law for $\widehat D_{K,t}$
between \emph{overlapping} windows that share most of their data and show that it
transfers without change to rotation-equivariant shrinkage estimators.
A distribution-free Davis--Kahan band gauges whether the eigenspace is
identified, an estimator-aware bootstrap provides the calibrated test, and a
companion power analysis gives an approximate design rule for the smallest
detectable rotation.
For the scalar functionals, we show that first-order immunity to elliptical kurtosis
holds for scale-invariant functionals and only for them, so that one estimated
scalar calibrates the projector null and the absorption-ratio and leading-share
intervals across the elliptical family.
In high dimensions, where shrinkage cleaning biases the absorption ratio, we give
a trace-preserving spike-debiased estimator that removes the bias.
The results are verified by simulation under a known population covariance; an
equity-panel appendix shows the procedures as diagnostics when the population is
unknown.
\end{abstract}


\newpage
\setcounter{tocdepth}{2}
\tableofcontents

\newpage
\section{Introduction}
\label{sec:intro}

A covariance matrix estimated on a rolling window is re-read continually for the
structure it encodes. Its leading eigenvalues measure how concentrated risk is,
the share of variance they carry is tracked as a barometer of systemic
co-movement \citep{Kritzman2011}, and the subspace spanned by its top
eigenvectors is taken to be the set of dominant factors. When these quantities
move, the change is interpreted as a shift in market structure and acted on. The
covariance is never observed, however. It is estimated, increasingly through
shrinkage or cleaning maps that improve matrix risk
\citep{LedoitWolf2004,LedoitWolf2020,LedoitWolf2022QIS}, and the eigenstructure of
an estimate drifts under estimation noise even when the population is fixed.
Separating genuine structural change from estimation noise therefore requires
knowing how much movement the noise alone produces, and how the cleaning map
reshapes it.

We take up that question here. When a window sample covariance $S_t$ is mapped
to an estimate $C_t$ (shrinkage, cleaning, or the identity $C_t=S_t$), how large
are the errors it induces in the nonlinear spectral functionals and in the
sequential top-$K$ projector that a monitor reports? Ledoit--Wolf theory bounds
matrix-level shrinkage risk
\citep{LedoitWolf2004,LedoitWolf2020,LedoitWolf2022QIS}; we instead condition on
the operator-norm error $\eta_t=\|C_t-\Sigma_t\|_{\op}$ and propagate it to the
spectral functionals and projector monitors through perturbation theory
\citep{DavisKahan1970,YuWangSamworth2015,Bhatia1997}. The propagation is
deterministic and is verified by finite-sample simulation under a known $\Sigma$,
with a real-data illustration when $\Sigma_t$ is unknown \citep{FanJiaoYao2018}.

Consider the dominant eigenspace first. When the gap between the $K$th and
$(K{+}1)$th population eigenvalues is wide, the top-$K$ eigenspace is a
well-identified feature of the population, and a genuine change in it signals a
change in the factor structure worth detecting. Its movement between consecutive
dates is therefore the natural quantity to monitor, and we measure it by the
Frobenius distance between the top-$K$ eigenprojectors estimated at $t-1$ and at
$t$. The observed movement mixes this genuine movement with estimation noise,
whose size is governed by the ratio of the estimation error to that same eigengap.
A wide gap thus makes observed movement informative, whereas a narrow gap lets
small errors produce large apparent rotation.

A deterministic, gap-adjusted bound makes the separation precise, controlling the
discrepancy between observed and latent movement by a quantity of exactly this
error-over-gap form and capping it at the largest distance two rank-$K$ projectors
can attain, so that it never exceeds the trivial bound
(Theorem~\ref{thm:dynamic-projector}).

The bound is worst-case and, in the dimensions typical of return panels, too
conservative to serve as a test, so we calibrate the monitor against an explicit
sampling model. The distribution of observed movement under no population change
follows from a first-order analysis between \emph{overlapping} windows that share
most of their data (Proposition~\ref{prop:noise-floor}), and the same calibration
serves the widely used rotation-equivariant shrinkage estimators without
modification (\S\ref{sec:rot-equiv}). A companion power analysis
(Proposition~\ref{prop:tier-b-power}) turns the detection limits of two-window
monitoring into closed-form expressions and, under a median-shift approximation,
into an approximate design rule for the smallest rotation a monitor can detect at
a target power and the number of fresh observations it requires. The analytic null
is exact only near the ends of the overlap range, and an estimator-aware
parametric bootstrap covers intermediate overlap, clipping-based estimators, and
heavy tails.

First-order immunity to elliptical kurtosis holds for scale-invariant spectral
functionals, and for them alone; the sufficiency direction is classical
\citep{Tyler1983,ShapiroBrowne1987}, while we establish the converse and the exact
finite-sample window-scale pivotality of Proposition~\ref{prop:exact-pivot}
(Proposition~\ref{prop:scale-immunity}). This immunity calibrates the two scalar
functionals the monitor reports. Proposition~\ref{prop:ar-clt} gives a fixed-$N$
elliptical central limit theorem with confidence intervals for the absorption
ratio and the leading share whose only elliptical input is a single estimated
scalar, and the same scalar calibrates the projector null. Spectral entropy is an
exception, its first-order intervals being unreliable near flat spectra, and we
report it as such.

The absorption ratio needs one further caution in high dimensions.
Matrix-shrinkage rankings and ratio-functional error need not agree
(Propositions~\ref{prop:ar-expansion} and~\ref{prop:matrix-vs-ar}), and in the
spiked proportional regime $N/M\to c$ the sample and Frobenius-optimal plug-in
absorption ratios acquire explicit, oppositely signed relative biases
(Proposition~\ref{prop:wedge}) that the trace-preserving spike-debiased estimator
of \S\ref{sec:ar-hd-caution} removes (Figure~\ref{fig:ar-wedge}).
Proposition~\ref{prop:propagate} records the complementary conditional propagation
of $\eta_t$ to the functionals.

The remainder of the paper is organized as follows.
Section~\ref{sec:setting} fixes notation and states the monitoring procedure.
Section~\ref{sec:theory} develops the worst-case perturbation bounds.
Section~\ref{sec:inference} provides the calibrated inference for eigenspace
movement, covering the overlapping-window null, its transfer to shrinkage
estimators, the bootstrap, and the power frontier.
Section~\ref{sec:functional-inference} calibrates the scalar functionals that the
monitor reports, with the kurtosis-immunity principle at its center and a
high-dimensional caution for the absorption ratio.
Section~\ref{sec:mc} collects all simulation evidence, and
Section~\ref{sec:discussion} concludes.
Proofs are collected in Appendix~\ref{app:proofs-detailed} and the equity-panel
illustration in Appendix~\ref{app:panel}.

\section{Setting and notation}
\label{sec:setting}

At calendar date $t$, let $r_{t-M+1},\ldots,r_t\in\R^N$ be the observations in a
window of length $M$, and let $\tilde r_\tau=r_\tau-\bar r$ be the deviations from
the window mean $\bar r$.
The sample covariance is $S_t=(M-1)^{-1}\sum_{\tau}\tilde r_\tau\tilde r_\tau^\top$.
We assume throughout that the estimator $\mathcal{A}$ is \emph{rotation-equivariant}, meaning
it maps the sample matrix to $C_t=\mathcal{A}(S_t)$ by acting on the eigenvalues of
$S_t$ alone and leaving its eigenvectors unchanged. The estimators studied are quadratic-inverse shrinkage (QIS)
\citep{LedoitWolf2022QIS}, linear shrinkage (LW) \citep{LedoitWolf2004},
oracle-approximating shrinkage (OAS) \citep{ChenWieselHeroEldar2010}, the sample
covariance itself, and Marchenko--Pastur eigenvalue clipping with single-spike
debiasing, denoted MP-clip \citep{BunBouchaudPotters2017,AhnHorenstein2013}; other
rotation-equivariant choices plug into the same pipeline.
The population matrix $\Sigma_t$ is the covariance of the data-generating
process.
Writing the eigenvalues of a matrix $C$ in decreasing order
$\lambda_1(C)\ge\cdots\ge\lambda_N(C)$, we let $U_K(C)$ collect its top-$K$
eigenvectors and write $P_{K,t}=U_K(\Sigma_t)U_K(\Sigma_t)^\top$ and
$\widehat P_{K,t}=U_K(C_t)U_K(C_t)^\top$ for the orthogonal projectors onto the span
of the top-$K$ eigenvectors of $\Sigma_t$ and $C_t$.
The monitored statistic is the movement of this projector between consecutive
dates,
\begin{equation}
  D_{K,t}=\|P_{K,t}-P_{K,t-1}\|_F,
  \qquad
  \widehat D_{K,t}=\|\widehat P_{K,t}-\widehat P_{K,t-1}\|_F,
  \label{eq:D-def-intro}
\end{equation}
the Frobenius distance between successive projectors; $D_{K,t}$ is the latent
movement and $\widehat D_{K,t}$ its observed counterpart.

Two derived quantities organize the analysis.
The first is the perturbation size $\eta_t=\|C_t-\Sigma_t\|_{\op}$, the largest
absolute eigenvalue of the error matrix $C_t-\Sigma_t$, equivalently the
worst-case error the estimate makes in any single direction. The bounds of
\S\ref{sec:theory} depend on the error only through $\eta_t$.
When $\Sigma_t$ is unobserved, $\eta_t$ is replaced by the sample-relative
quantity $\|C_t-S_t\|_{\op}$.
The second is the eigengap at the monitoring cut,
\begin{equation}
  \Delta_{K,t}=\lambda_K(\Sigma_t)-\lambda_{K+1}(\Sigma_t),
  \label{eq:eigengap}
\end{equation}
which has the same units as $\eta_t$ and sets how strongly an eigenvalue
perturbation can rotate the top-$K$ eigenspace, with a wide gap pinning it
down and a narrow gap leaving it loosely determined.

\paragraph{Spectral functionals.}
\begin{align*}
  f_1(C) &= \lambda_1(C)\Big/\textstyle\sum_j \lambda_j(C), \\
  \AR_K(C) &= \textstyle\sum_{i\le K}\lambda_i(C)\Big/\textstyle\sum_j \lambda_j(C), \\
  H(C) &= -\textstyle\sum_i p_i\log p_i,\quad p_i=\lambda_i(C)/\textstyle\sum_j\lambda_j(C).
\end{align*}
Each is invariant under rescaling the covariance, $W(bC)=W(C)$ for $b>0$,
because the numerator and denominator scale together; this scale invariance is
the basis of the calibration theory of \S\ref{sec:functional-inference}.
The monitoring band depends only on the eigengap $\Delta_{K,t}$ of \eqref{eq:eigengap}.
Some empirical results are additionally stratified by descriptive quantities,
the normalized gap $\gaptwelve=(\lambda_1-\lambda_2)/\lambda_1$, the angle
between successive leading eigenvectors, and the chordal distance between
successive eigenspaces, but none of these enters the band.

\subsection{Monitoring procedure}
\label{sec:algorithm}

Algorithm~\ref{alg:spectral-monitor} is the pipeline used in \S\ref{sec:mc} and
Appendix~\ref{app:panel}.
It computes the projector movement of \eqref{eq:D-def-intro} and the band of
Theorem~\ref{thm:dynamic-projector}, using the population error $\eta_t$ in
simulation and the sample-relative error $\|C_t-S_t\|_{\op}$ on real data; the
bound itself is unchanged by the substitution.

\begin{algorithm}[!ht]
\caption{Sequential gap-adjusted spectral monitor at date $t$.}
\label{alg:spectral-monitor}
\hrule\smallskip
\noindent\textbf{Input:} return window $W_t$; covariance estimator
$\mathcal{A}$; maximum dimension $K_{\max}$; optional scale $\alpha\in(0,1]$.
\smallskip\hrule\smallskip
\begin{enumerate}\setlength{\itemsep}{2pt}
  \item $S_t\gets$ sample covariance of $W_t$;\quad $C_t\gets\mathcal{A}(S_t)$.
  \item Eigendecompose $C_t$ and select $\widehat K_t\le K_{\max}$ (largest
  eigengap, or an $\AR_k$ threshold).
  \item $\widehat P_{K,t}\gets U_{\widehat K_t}U_{\widehat K_t}^\top$;\quad
  $\widehat D_{K,t}\gets\|\widehat P_{K,t}-\widehat P_{K,t-1}\|_F$.
  \item $\eta_t\gets\|C_t-\Sigma_t\|_{\op}$ in simulation, or
  $\|C_t-S_t\|_{\op}$ on real data.
  \item Band $\tau^*_{K,t}\gets\min\!\bigl(\sqrt{2K},\;T_{K,t}+T_{K,t-1}\bigr)$, with
  per-date term $T_{K,s}=\min\!\bigl(2\sqrt{2}\,\sqrt{K}\,\eta_s/\Delta_{K,s},\;\sqrt{2K}\bigr)$
  of \eqref{eq:dk-term} capped at the maximum projector distance $\sqrt{2K}$; on real
  data replace the population gap $\Delta_{K,s}$ by the observable
  $(\widehat\Delta_{K,s}-2\eta_s)_+$.
  \item If $\widehat K_t\ne\widehat K_{t-1}$, record a rank-change event
  (\S\ref{sec:theory-monitor}).
\end{enumerate}
\smallskip\hrule\smallskip
\noindent\textbf{Output:} $C_t$; the functionals $f_1,\AR_K,H$; the movement
$\widehat D_{K,t}$; the band $\tau^*_{K,t}$; the flag
$\mathbb{1}\{\widehat D_{K,t}>\alpha\tau^*_{K,t}\}$.
\smallskip\hrule
\end{algorithm}

\noindent The estimator $\mathcal{A}$ is QIS, LW, MP-clip, or the identity (the
sample covariance $S_t$).

\newpage
\section{Worst-case perturbation bounds}
\label{sec:theory}

All estimation error enters the monitored objects through the operator-norm
budget $\eta_t=\|C_t-\Sigma_t\|_{\op}$, and we first ask what this budget alone
implies. Under the single assumption $\|C_t-\Sigma_t\|_{\op}\le\eta_t$,
Theorem~\ref{thm:dynamic-projector} bounds the deviation of the observed top-$K$
projector movement from the latent movement in terms of $\eta_t$ and the eigengap.
The bound holds for every estimator and every return distribution and is
correspondingly conservative.

\subsection{A band for sequential eigenspace monitoring}
\label{sec:theory-monitor}

Only the estimated projector $\widehat P_{K,t}$ is observed, so the observed
movement $\widehat D_{K,t}$ of \eqref{eq:D-def-intro} differs from the latent
movement $D_{K,t}$. We bound this difference by first bounding the per-date
deviation of an estimated projector from its population target, and then combining
two consecutive dates through the reverse triangle inequality. The estimator
enters only through the bound $\|C_s-\Sigma_s\|_{\op}\le\eta_s$.

The projector distance is determined by the principal angles between the two
subspaces.

\begin{lemma}[Projector--$\sin\Theta$ identity]
\label{lem:proj-sin}
For rank-$K$ projectors $P_K=U_KU_K^\top$ and
$\widehat P_K=\widehat U_K\widehat U_K^\top$,
\[
  \|P_K-\widehat P_K\|_F=\sqrt{2}\,\|\sin\Theta(U_K,\widehat U_K)\|_F
\]
\citep{StewartSun1990,YuWangSamworth2015}.
\end{lemma}

The Davis--Kahan theorem bounds these angles by the perturbation size relative to
the eigengap $\Delta_{K,t}$, the bound deteriorating as $\Delta_{K,t}\to 0$.

\begin{lemma}[Davis--Kahan $\sin\Theta$ bound]
\label{lem:dk-sin}
Let $\Delta_{K,t}=\lambda_K(\Sigma_t)-\lambda_{K+1}(\Sigma_t)>0$ and
$\|E_t\|_{\op}\le\eta_t$, where $E_t=C_t-\Sigma_t$. Then
\begin{equation}
  \|\sin\Theta(\widehat U_{K,t},U_{K,t})\|_F
  \le \frac{2\min\bigl(\sqrt{K}\,\eta_t,\|E_t\|_F\bigr)}{\Delta_{K,t}}
  \le \frac{2\sqrt{K}\,\eta_t}{\Delta_{K,t}}.
  \label{eq:dk-sin}
\end{equation}
\end{lemma}

\noindent Inequality~\eqref{eq:dk-sin} is Theorem~2 of \citet{YuWangSamworth2015}, the
population-gap Davis--Kahan variant, specialized to the top-$K$ block ($r=1$, $s=K$).

The right-hand side of \eqref{eq:dk-sin} diverges as $\Delta_{K,t}\to 0$, whereas the
projector distance never exceeds $\sqrt{2K}$; the latter furnishes an unconditional cap.

\begin{lemma}[Maximum projector distance]
\label{lem:trivial-cap}
For rank-$K$ orthogonal projectors $P,\widehat P$,
$\|P-\widehat P\|_F^2=2K-2\langle P,\widehat P\rangle\le 2K$, since
$\langle P,\widehat P\rangle=\mathrm{tr}(P\widehat P)\ge 0$.
Hence $\|P-\widehat P\|_F\le\sqrt{2K}$ always, the largest Frobenius distance two
rank-$K$ projectors can attain, reached exactly when their ranges are orthogonal
($\langle P,\widehat P\rangle=0$). More generally $\|P_a-P_b\|_F\le\sqrt{a+b}$
for ranks $a,b$.
\end{lemma}

Combining the angle identity, the Davis--Kahan bound, and the cap yields the
per-date deviation $\|\widehat P_{K,t}-P_{K,t}\|_F$.

\begin{lemma}[Per-date projector deviation]
\label{lem:dk-proj}
Under the assumptions of Lemma~\ref{lem:dk-sin}, define the per-date band term
\begin{equation}
  T_{K,t}
  \;=\;
  \min\!\Bigl(
  \frac{2\sqrt{2}\,\min\bigl(\sqrt{K}\,\eta_t,\|E_t\|_F\bigr)}{\Delta_{K,t}},
  \;\sqrt{2K}
  \Bigr).
  \label{eq:dk-term}
\end{equation}
Then
\[
  \|\widehat P_{K,t}-P_{K,t}\|_F
  \le T_{K,t}
  \le \frac{2\sqrt{2K}\,\eta_t}{\Delta_{K,t}}.
\]
\end{lemma}

\begin{proof}
Lemma~\ref{lem:proj-sin} gives
$\|\widehat P_{K,t}-P_{K,t}\|_F=\sqrt{2}\,\|\sin\Theta(\widehat U_{K,t},U_{K,t})\|_F$,
and Lemma~\ref{lem:dk-sin} bounds
$\|\sin\Theta\|_F\le 2\min(\sqrt{K}\,\eta_t,\|E_t\|_F)/\Delta_{K,t}$;
multiplying the two yields the first argument of the minimum in \eqref{eq:dk-term}.
Lemma~\ref{lem:trivial-cap} supplies the second argument $\sqrt{2K}$, which holds with
no gap assumption at all. The final inequality follows from
$\min(\sqrt K\eta_t,\|E_t\|_F)\le\sqrt K\eta_t$ and $\min(x,\sqrt{2K})\le x$.
\end{proof}

Applying the per-date deviation at $t-1$ and $t$ and combining the two by the
reverse triangle inequality yields the movement bound.

\begin{theorem}[Gap-adjusted monitoring band]
\label{thm:dynamic-projector}
Fix the monitoring dimension $K$, and let $D_{K,t}$ and $\widehat D_{K,t}$ be the
latent and observed projector movements of \eqref{eq:D-def-intro}.
Suppose the eigengaps are positive, $\Delta_{K,t}>0$ and $\Delta_{K,t-1}>0$, and
the estimation errors are controlled, $\|C_s-\Sigma_s\|_{\op}\le\eta_s$ for
$s\in\{t-1,t\}$.
Then, with the per-date term $T_{K,s}$ of \eqref{eq:dk-term},
\begin{equation}
  \bigl|\widehat D_{K,t}-D_{K,t}\bigr|
  \le \tau^*_{K,t}
  :=\min\Bigl(\sqrt{2K},\;T_{K,t}+T_{K,t-1}\Bigr)
  \le 2\sqrt{2K}\Bigl(\frac{\eta_t}{\Delta_{K,t}}+\frac{\eta_{t-1}}{\Delta_{K,t-1}}\Bigr).
  \label{eq:dynamic-D}
\end{equation}
The outer cap $\sqrt{2K}$ requires no gap condition (Lemma~\ref{lem:trivial-cap});
the uncapped right-hand bound uses $c_{K}=2\sqrt{2K}$, i.e.\ $c_1=2\sqrt2$, $c_2=4$.
\end{theorem}

\noindent The gap-adjusted budget $\tau^*_{K,t}$ is the most that estimation noise of
size $\eta_s$ relative to the eigengap $\Delta_{K,s}$ can shift the observed movement
away from the latent movement.
The bound is worst-case and, capped at $\sqrt{2K}$, certifies only movement too large
to be noise; it is an identification gauge rather than a calibrated test, and is
typically loose in the high-dimensional regimes of \S\ref{sec:theory-prop}, which is
why \S\ref{sec:inference} calibrates it against an explicit sampling model.

\begin{proof}
Write $X=\widehat P_{K,t}-\widehat P_{K,t-1}$ and $Y=P_{K,t}-P_{K,t-1}$.
The reverse triangle inequality gives
$|\widehat D_{K,t}-D_{K,t}|=|\|X\|_F-\|Y\|_F|\le \|X-Y\|_F$, and
\[
  X-Y
  =(\widehat P_{K,t}-P_{K,t})-(\widehat P_{K,t-1}-P_{K,t-1}),
\]
so the triangle inequality yields
\[
  \|X-Y\|_F
  \le \|\widehat P_{K,t}-P_{K,t}\|_F
  +\|\widehat P_{K,t-1}-P_{K,t-1}\|_F.
\]
Applying Lemma~\ref{lem:dk-proj} at $t$ and $t-1$ and adding the two bounds gives
$\|X-Y\|_F\le T_{K,t}+T_{K,t-1}$.
Independently, $D_{K,t},\widehat D_{K,t}\in[0,\sqrt{2K}]$ by
Lemma~\ref{lem:trivial-cap}, so $|\widehat D_{K,t}-D_{K,t}|\le\sqrt{2K}$ with no gap
condition. Taking the minimum of the two bounds produces $\tau^*_{K,t}$ in
\eqref{eq:dynamic-D}; the final inequality follows from
$T_{K,s}\le 2\sqrt{2K}\,\eta_s/\Delta_{K,s}$.
\end{proof}

The monitoring dimension need not be held fixed across dates. If the selected cut
differs ($K_t\ne K_{t-1}$), the same argument applies with each date's own cut,
giving
$|\widehat D_{t}-D_{t}|\le\min\bigl(\sqrt{K_t+K_{t-1}},\,T_{K_t,t}+T_{K_{t-1},t-1}\bigr)$,
where $D_t$ and $\widehat D_t$ now compare a rank-$K_t$ projector with a
rank-$K_{t-1}$ one and Lemma~\ref{lem:trivial-cap} supplies the cap
$\sqrt{K_t+K_{t-1}}$. A rank change is itself a form of movement, since
$\|P_a-P_b\|_F^2\ge|a-b|$ forces a distance of at least $\sqrt{|K_t-K_{t-1}|}$, so
we record rank-change dates separately rather than as rotation of a
fixed subspace.

\begin{corollary}[Diagnostic flag]
\label{cor:alarm}
Under the conditions of Theorem~\ref{thm:dynamic-projector}, if
$\widehat D_{K,t}>\tau^*_{K,t}$ then $D_{K,t}\ne 0$, and the population eigenspace
must have moved for the observed movement to exceed the band.
\end{corollary}

\begin{proof}
If $D_{K,t}=0$ then $P_{K,t}=P_{K,t-1}$, and the quantity $X-Y$ in the proof of
Theorem~\ref{thm:dynamic-projector} reduces to
$\widehat P_{K,t}-\widehat P_{K,t-1}$, so the theorem gives
$\widehat D_{K,t}\le\tau^*_{K,t}$. The contrapositive is the claim.
The implication is diagnostic. It identifies movement too
large to be perturbation.
\end{proof}

\noindent The flag inherits its guarantee from the worst-case construction. The band
$\tau^*_{K,t}$ is the population-gap Davis--Kahan bound of \S\ref{sec:theory-monitor}
(Lemma~\ref{lem:dk-sin}) carried across two consecutive dates and capped at the maximum
projector distance $\sqrt{2K}$.

Corollary~\ref{cor:alarm} uses the population error
$\eta_t=\|C_t-\Sigma_t\|_{\op}$, which is available only in simulation, where
$\Sigma_t$ is known by construction. On real data $\Sigma_t$ is unobserved and the
error must instead be measured against the window sample covariance,
$\|C_t-S_t\|_{\op}$, so the conclusion weakens accordingly. An exceedance
$\widehat D_{K,t}>\tau^*_{K,t}$ then shows only that the cleaned and raw
projectors disagree by more than cleaning noise allows, which is a diagnostic
about the estimator rather than proof that the population eigenspace moved.
Establishing $D_{K,t}\ne 0$ in that case requires the external calibration of
\S\ref{sec:inference}, and the equity panel of Appendix~\ref{app:panel} runs the
monitor in exactly this sample-relative mode.

The per-date term $T_{K,s}$ of \eqref{eq:dk-term} still contains the population
gap $\Delta_{K,s}$, which is not observable on real data. By Weyl's inequality
(Lemma~\ref{lem:weyl}) each eigenvalue moves by at most $\eta_s$, so the estimated
gap $\widehat\Delta_{K,s}=\lambda_K(C_s)-\lambda_{K+1}(C_s)$ satisfies
$|\widehat\Delta_{K,s}-\Delta_{K,s}|\le 2\eta_s$, hence
$\Delta_{K,s}\ge\widehat\Delta_{K,s}-2\eta_s$. Substituting this lower bound for
$\Delta_{K,s}$ preserves the guarantee.

\begin{corollary}[Observable band]
\label{cor:empirical-gap}
Define $\widehat T_{K,s}$ by substituting $(\widehat\Delta_{K,s}-2\eta_s)_+$ for
$\Delta_{K,s}$ in \eqref{eq:dk-term}, with $\widehat T_{K,s}=\sqrt{2K}$ when
$\widehat\Delta_{K,s}\le 2\eta_s$, and set
$\widehat\tau^*_{K,t}=\min\bigl(\sqrt{2K},\;\widehat T_{K,t}+\widehat T_{K,t-1}\bigr)$.
Then $\widehat\tau^*_{K,t}\ge\tau^*_{K,t}$, the band $\widehat\tau^*_{K,t}$ is
computable from $(C_s,\eta_s)$ alone, and Theorem~\ref{thm:dynamic-projector} and
Corollary~\ref{cor:alarm} hold with $\widehat\tau^*_{K,t}$ in place of
$\tau^*_{K,t}$.
\end{corollary}

Corollary~\ref{cor:empirical-gap} is the form of the band deployed on real data,
where $\Sigma_t$, and with it $\Delta_{K,s}$ and $\tau^*_{K,t}$, is unavailable
(Algorithm~\ref{alg:spectral-monitor}, step~5, and Appendix~\ref{app:panel});
the simulations of \S\ref{sec:mc}, where $\Sigma$ is known, use the population
band directly. Enlarging the band preserves the one-sided guarantee of
Corollary~\ref{cor:alarm} at the price of additional conservatism, the
worst-case cost of not observing the gap.

\begin{remark}[$\alpha$-scaled conservative band]
\label{rem:alpha-scaled-tau}
For $\alpha\in(0,1]$, define the \emph{scaled} diagnostic
$\widehat D_{K,t}>\alpha\,\tau^*_{K,t}$.
The scaled band is a conservative rescaling of the Davis--Kahan \emph{upper} bound in
Corollary~\ref{cor:alarm}, recovered at $\alpha{=}1$.
The simulations vary $\alpha$ to quantify the conservatism of the band
(Table~\ref{tab:mc-alpha}); because the flag rate decreases in $\alpha$,
calibration selects the smallest $\alpha^\star$ on the grid at which the
simulated flag rate under the null stays at or below $5\%$, and the same scale
is used on real data.
At $\alpha{=}1$ the band is the worst-case bound itself; the analytic test of
\S\ref{sec:inference} provides an alternative calibration when the band at
$\alpha{=}1$ never flags.
\end{remark}

\begin{lemma}[Weyl's inequality]
\label{lem:weyl}
Let $C,\Sigma\in\R^{N\times N}$ be symmetric with eigenvalues
$\lambda_1(C)\ge\cdots\ge\lambda_N(C)$ and $\lambda_i(\Sigma)$ ordered likewise,
and let $\eta=\|C-\Sigma\|_{\op}$. Then for every $i\in\{1,\ldots,N\}$,
\[
  |\lambda_i(C)-\lambda_i(\Sigma)|\le \eta
\]
\citep[Cor.~VI.1.6]{StewartSun1990}.
\end{lemma}

\setcounter{proposition}{0}
\begin{proposition}[Stable monitoring dimension]
\label{prop:stable-K}
Fix $K_{\max}\ge 1$. Assume the population argmax $K_t^\star=\arg\max_{1\le k\le K_{\max}}\Delta_{k,t}$ is unique; if not,
fix a deterministic tie-breaking rule for both $K_t^\star$ and $\widehat K_t$.
Let $\widehat K_t=\arg\max_{1\le k\le K_{\max}}\widehat\Delta_{k,t}$.
If $\|C_t-\Sigma_t\|_{\op}\le\varepsilon_t$ and the margin
\begin{equation}
  \Delta_{K_t^\star,t}-\max_{1\le k\le K_{\max},\,k\ne K_t^\star}\Delta_{k,t}>4\varepsilon_t,
  \label{eq:k-margin}
\end{equation}
holds, then $\widehat K_t=K_t^\star$.
We interpret $\widehat K_t$ as the \emph{most stable monitoring dimension}
(dominant eigengap cut) rather than as an estimate of the true factor rank.
\end{proposition}

\begin{proof}
Fix $k\ne K_t^\star$. The margin \eqref{eq:k-margin} is strict, so
$\Delta_{K_t^\star,t}>\Delta_{k,t}+4\varepsilon_t$.
Weyl's inequality (Lemma~\ref{lem:weyl}) moves each eigenvalue by at most
$\varepsilon_t$ and so each eigengap by at most $2\varepsilon_t$, giving
$\widehat\Delta_{K_t^\star,t}\ge\Delta_{K_t^\star,t}-2\varepsilon_t
>\Delta_{k,t}+2\varepsilon_t$, while
$\widehat\Delta_{k,t}\le\Delta_{k,t}+2\varepsilon_t$.
Hence $\widehat\Delta_{K_t^\star,t}>\widehat\Delta_{k,t}$ for every $k\ne K_t^\star$,
and $\widehat K_t=K_t^\star$ under the tie-breaking convention.
\end{proof}

\begin{remark}[The monitoring constant and the cap]
\label{rem:cdk}
The constant $c_K=2\sqrt{2K}$ in the uncapped bound of
Theorem~\ref{thm:dynamic-projector} follows from
Lemmas~\ref{lem:proj-sin} and~\ref{lem:dk-sin} ($c_1=2\sqrt2\approx 2.83$,
$c_2=4$); the constant is $K$-dependent and is computed from the selected cut.
The cap matters in practice. The term $2\sqrt{2K}\,\eta_s/\Delta_{K,s}$ exceeds the maximum possible
projector distance $\sqrt{2K}$ as soon as $\eta_s/\Delta_{K,s}>1/2$, so in
weak-identification regimes the uncapped band is vacuous while $\tau^*_{K,t}$ saturates at
its ceiling and carries no information beyond the trivial bound.
\end{remark}

\subsection{Regimes for the operator-norm budget}
\label{sec:theory-prop}

The band $\tau^*_{K,t}$ of Theorem~\ref{thm:dynamic-projector} is informative only
when the budget $\eta_t$ is small relative to the eigengap. Two asymptotic regimes
determine when this holds. When $N$ and $M$ are comparable (the return-panel
regime), the budget stays $O(1)$ and the band only gauges whether the eigenspace is
identified; when $N$ is fixed and $M$ grows the budget shrinks at the root-$M$ rate and
the band tightens into a usable bound.

The first is the proportional regime, in which the dimension $N$ and the window
length $M$ are comparable. Here the budget does not vanish. If $N,M\to\infty$ with
$N/M\to c\in(0,\infty)$, then $\|S_t-\Sigma_t\|_{\op}=O_p(\sqrt{N/M})=O_p(1)$ does
not converge to zero, because the Marchenko--Pastur limit leaves the sample
spectrum persistently distorted
\citep{MarchenkoPastur1967,BaiYin1993,Vershynin2018}; in particular $\sqrt{N/M}$
is not $O_p(M^{-1/2})$ once $c>0$. Return panels live in this regime, with the
cross-section a sizeable fraction of the window length, as in the equity panel
of Appendix~\ref{app:panel} at $N/M\approx0.46$. Operator-norm consistency there
demands extra structure such as sparsity \citep{ElKaroui2008} or a spiked model
\citep{Johnstone2001,FanJiaoYao2018}, and shrinkage or cleaning, while it improves
matrix risk in Frobenius or oracle terms
\citep{LedoitWolf2020,LedoitWolf2022QIS,BunBouchaudPotters2017,DonohoGavish2014},
need not drive $\|C_t-\Sigma_t\|_{\op}$ to zero. In this regime the band serves
only as an identification gauge. The simulation studies accordingly evaluate
functional error against a known $\Sigma$, and on real data $\|C_t-S_t\|_{\op}$ is
a sample-relative diagnostic rather than a statement about population accuracy.

The budget does shrink in the second, fixed-dimension regime. If $N$ is fixed and
$M\to\infty$, then $\|S_t-\Sigma_t\|_{\op}=O_p(M^{-1/2})$, and under Ledoit--Wolf
conditions $\|C_t-\Sigma_t\|_{\op}=O_p(M^{-1/2}+\delta_t)$ with
$\delta_t=O_p(M^{-1})$, so the band tightens at the usual root-$M$ rate. The
calibrated tests of \S\ref{sec:inference} are derived in this regime.

\newpage
\section{Calibrated inference for eigenspace movement}
\label{sec:inference}

The band $\tau^*_{K,t}$ of \S\ref{sec:theory} assumes only a bound on $\eta_t$. In
the proportional regime of \S\ref{sec:theory-prop} its per-date terms saturate the
cap $\sqrt{2K}$, so it acts as an identification gauge rather than a detector
(Remark~\ref{rem:cdk}; Table~\ref{tab:tier-power}). Attaching a significance level
to an observed movement $\widehat D_{K,t}$ requires its distribution under no
population change, and hence a sampling model.

\subsection{Reduction to the sample-covariance null}
\label{sec:rot-equiv}

In principle a null distribution for $\widehat D_{K,t}$ must be derived for each
estimator $\mathcal{A}$ separately, since each reshapes the spectrum differently.
Most do not require this. A rotation-equivariant map
$C_t=U\,\mathrm{diag}(\delta(\lambda))\,U^\top$ acts only on the eigenvalues of
$S_t=U\,\mathrm{diag}(\lambda)\,U^\top$ and leaves the eigenbasis $U$ unchanged.
Provided $\delta$ preserves the order of the top $K$ eigenvalues, the top-$K$
projector of $C_t$ equals that of $S_t$, so $\widehat D_{K,t}$ is identical across
all such estimators, among them linear shrinkage, QIS, and eigenvalue clipping.

The estimators differ in whether they meet this condition. QIS and linear
shrinkage preserve the ordering to machine precision
($\|P_K(C_t)-P_K(S_t)\|_F<10^{-14}$ across simulated windows). Marchenko--Pastur
clipping does not, because averaging the bulk creates eigenvalue ties; the
eigenbasis within a tied block is then arbitrary and the projector can rotate by
$O(1)$. The sample-covariance null therefore applies without change to QIS and
linear-shrinkage monitors, while clipping-based monitors, lacking the identity,
are calibrated by simulation instead.

\subsection{A first-order null for overlapping windows}
\label{sec:tier-b}

Successive evaluation dates overlap, with the windows at $t-1$ and $t$ sharing
$M-s$ of their $M$ observations, where $s$ is the stride between dates ($s=1$ for
daily monitoring). We write $\rho=M_{\mathrm{overlap}}/M=(M-s)/M=1-\phi$ for the
overlap fraction and $\phi=s/M$ for the fresh fraction. The first-order null below is
the high-overlap regime ($\rho\to1$, i.e.\ $\phi\to0$), exact in that limit and
at the disjoint end $\rho=0$, with a documented size drift at intermediate overlap
($\rho\approx0.75$, i.e.\ $s/M\approx0.25$; Figure~\ref{fig:null-vs-chi2}).
Because the shared observations enter both window covariances in the same way,
they cancel in the difference $S_t-S_{t-1}$, which is therefore driven entirely
by the $2s$ fresh observations.

Calibrated inference for $\|\widehat P-P\|_F$ from a \emph{single} sample is
classical \citep{Anderson1963,KoltchinskiiLounici2017}, with bootstrap
confidence sets \citep{NaumovSpokoinyUlyanov2019} and dimension-free extensions
\citep{JirakWahl2023,JirakWahl2024}, while two-sample eigenspace tests assume
\emph{independent} samples \citep{Schott1988}. Neither setting
addresses two windows that share part of their data.
For overlapping windows, the econophysics literature
\citep{AllezBouchaud2012,BunBouchaudPotters2018} uses random-matrix overlap
curves as informal nulls but supplies neither fluctuation theory nor a test.
In the proportional regime, \citet{LinPan2024} characterize the limiting
eigenvector overlaps of Ledoit--Wolf nonlinear shrinkage and quantify its loss
through them, again as deterministic large-$(N,M)$ formulas rather than a
finite-sample test.
What has been missing is a calibrated treatment of sequential projector movement
between overlapping windows, and one that covers shrinkage estimators.

\begin{proposition}[First-order null for overlapping windows]
\label{prop:noise-floor}
Let $r_1,\ldots,r_{M+s}\overset{iid}{\sim}N(0,\Sigma)$ with simple eigenvalues
$\lambda_1>\cdots>\lambda_N$, fixed $N$, and let $\widehat P_{K,t}$, $\widehat
P_{K,t-1}$ be top-$K$ sample-covariance projectors of windows $\{1,\ldots,M\}$ and
$\{s+1,\ldots,s+M\}$.
Then, with $E$ the covariance increment, first-order perturbation
(\citealt{Anderson1963,KoltchinskiiLounici2017}) gives
\begin{equation}
  \E\,\widehat D_{K,t}^2
  =\frac{4s}{M^2}\sum_{i\le K<j}\frac{\lambda_i\lambda_j}{(\lambda_i-\lambda_j)^2}
  \;+\;o(s/M^2),
  \label{eq:noise-floor}
\end{equation}
and the first-order law of $\widehat D_{K,t}$ is that of
\begin{equation}
  \Bigl[\,2\sum_{i\le K<j}\Bigl(\frac{g_{ij}}{\lambda_i-\lambda_j}\Bigr)^2\Bigr]^{1/2},
  \qquad
  g_{ij}=\frac{1}{M}\Bigl(\sum_{k\,\mathrm{new}}y_{ik}y_{jk}-\sum_{k\,\mathrm{old}}y_{ik}y_{jk}\Bigr),
  \label{eq:exact-null}
\end{equation}
with $y_{ik}\sim N(0,\lambda_i)$ independent, a sum of $2s$ Gaussian products per
entry that is simulated exactly in $O(sN)$ operations per draw and needs no
eigendecomposition.
At $s=M$, where the windows are disjoint, \eqref{eq:noise-floor} equals exactly
twice the one-sample mean of \citet{KoltchinskiiLounici2017}.
For $s\to\infty$ the entries Gaussianize and
$(M^2/4s)\widehat D_{K,t}^2\to_d\sum_{i\le K<j}w_{ij}Z_{ij}^2$,
$w_{ij}=\lambda_i\lambda_j/(\lambda_i-\lambda_j)^2$, $Z_{ij}\overset{iid}{\sim}N(0,1)$.
The law is first order in the per-step increment, and the $o(s/M^2)$ remainder is bounded in
Proposition~\ref{prop:second-order}, so it is most accurate in the high-overlap regime
($\rho\to1$, small fresh fraction $\phi=s/M$) and at the disjoint end ($\rho=0$), with an
upward quantile drift at intermediate overlap.
\end{proposition}

\begin{remark}[Elliptical extension of the null]
\label{rem:elliptical-null}
Under elliptical sampling with kurtosis parameter $\kappa$,
$\mathrm{Var}(u_i^\top r r^\top u_j)=(1+\kappa)\lambda_i\lambda_j$ for $i\ne j$, so every
pair variance, and with it \eqref{eq:noise-floor}, scales by $(1+\kappa)$; null
draws of $\widehat D_{K,t}$ therefore scale by $\sqrt{1+\kappa}$.
Operationally we use $q^{\mathrm{an,ell}}_{95,t}=\sqrt{1+\widehat\kappa_t}\,q^{\mathrm{an}}_{95,t}$
with $\widehat\kappa_t$ from the radial-MLE estimator of \S\ref{sec:ar-inference}.
Corollary~\ref{cor:dhat-ell} shows this scaling is not an ad-hoc correction but the
projector instance of a general principle. The common-shock fourth-moment term
cannot load on any scale-invariant functional, and the top-$K$ projector is one.
\end{remark}

\noindent\textit{Proof.} Appendix~\ref{app:noise-floor} (perturbation expansion of
the projector difference; shared-window cancellation; variance computation).
In simulation, the mean \eqref{eq:noise-floor} is accurate to within
$5\%$ across $s/M\in\{0.08,0.25,1\}$, and the $95\%$ quantile of \eqref{eq:exact-null}
matches simulation at $s/M\in\{0.08,1\}$ with size $5.5\%$ under plug-in eigenvalues;
at intermediate overlap ($s/M=0.25$) second-order terms shared between windows inflate
the true quantile (empirical size $\approx 10\%$), the case the estimator-aware
bootstrap below is designed to handle.

The first-order law omits the quadratic remainder of the projector expansion.
Under the gap condition $\eta_s<\Delta/4$ the remainder is controlled by the
ratio $\eta_s/\Delta$ alone, and its size bounds both the deviation of
$\widehat D_{K,t}$ from the linear term and the resulting size distortion of
the analytic test.

\begin{proposition}[Second-order remainder bound]
\label{prop:second-order}
Let $\Sigma$ have top-$K$ gap $\Delta=\lambda_K-\lambda_{K+1}>0$, set
$r=(\lambda_1-\lambda_K+\Delta)/2$, and write
$L(\Delta E)$ for the linear term of Proposition~\ref{prop:noise-floor}.
If $\eta_s=\|E_s\|_{\op}<\Delta/4$ for $s\in\{t-1,t\}$, then
\begin{equation}
  \Bigl|\widehat D_{K,t}-\|L(\Delta E)\|_F\Bigr|
  \;\le\;
  \sum_{s\in\{t-1,t\}} 2\sqrt{K}\,\frac{2r}{\Delta}
  \Bigl(\frac{2\eta_s}{\Delta}\Bigr)^{2}\Bigl(1-\frac{2\eta_s}{\Delta}\Bigr)^{-1}
  \;\le\;
  \frac{64\sqrt{K}\,r\,(\eta_t\vee\eta_{t-1})^2}{\Delta^3}.
  \label{eq:second-order}
\end{equation}
Consequently, for any $h<\Delta/4$, the analytic test with the inflated threshold
$q^{\mathrm{an}}_{95}+64\sqrt{K}rh^2/\Delta^3$ has size at most
$0.05+\mathbb{P}(\eta_t\vee\eta_{t-1}\ge h)$ under the Gaussian null, where the
tail probability admits standard operator-norm concentration bounds in terms of the
effective rank \citep[][Bernoulli version]{KoltchinskiiLounici2017}.
\end{proposition}

\noindent\textit{Proof.} Appendix~\ref{app:second-order} (contour-integral Neumann
expansion; the remainder has rank at most $4K$, converting operator to Frobenius norm
at cost $2\sqrt K$).
The bound rests on a standard picture. The top-$K$ projector is a contour integral
of the resolvent $(zI-\Sigma)^{-1}$ around the top-$K$ eigenvalues, and perturbing
$\Sigma$ expands that resolvent in a Neumann (geometric) series in the error $E$; the
first term reproduces the linear law $L(\Delta E)$, and the remainder is the tail of the
series, hence quadratic in $\eta/\Delta$.
Across $1{,}200$ randomized trials spanning $\eta/\Delta\in(0.02,0.245)$ the
inequality was never violated, with median slack below $0.02$.

The remainder difference $Q(E_t)-Q(E_{t-1})$ does not inherit the shared-block
cancellation that the linear terms enjoy. Writing $E_s=E_c+E_s'$ (common block plus
$s$-specific block), the quadratic cross terms $Q_2(E_c,E_t'-E_{t-1}')$ survive.
Proposition~\ref{prop:second-order} is a worst-case envelope for this effect; at the
panel's $N/M$ the condition $\eta<\Delta/4$ typically fails and the bound is
uninformative, which is precisely when the estimator-aware bootstrap is the
recommended default.
The empirical size profile of the uncorrected analytic test ($5.5\%$ at $s/M=0.08$,
$10\%$ at $s/M=0.25$, $5.2\%$ at $s/M=1$) is consistent with this mechanism. The
shared-block cross term vanishes at $s=M$ and is masked by the heavy-tailed exact
increments at very small $s$.

\subsection{An estimator-aware parametric bootstrap}
\label{sec:tier-c}

The analytic null is exact only at the ends of the overlap range and assumes
rotation-equivariant cleaning. Two cases escape it. Its size drifts upward at
intermediate overlap, and for clipping-based estimators the projector identity of
\S\ref{sec:rot-equiv} fails. Both are handled by calibrating the threshold by
simulation.
At date $t$, simulate $B$ Gaussian paths of length $M+s$ from $N(0,C_t)$, apply the
\emph{same} estimator $\mathcal{A}$ to both windows of each path, and use the empirical
$95\%$ quantile $q^{\mathrm{bt}}_{95,t}$ of the simulated $\widehat D_{K,t}$ as the
threshold (Algorithm~\ref{alg:bootstrap}).
This captures finite-$s$ tails, second-order overlap terms, and any estimator
non-equivariance, at the cost of $2B$ estimator fits per evaluation date.
It adapts the parametric-bootstrap idea of \citet{NaumovSpokoinyUlyanov2019} from
one-sample confidence sets for $P$ to the two-window movement statistic with a
\emph{cleaned} plug-in $C_t$.

\begin{algorithm}[!ht]
\caption{Estimator-aware parametric bootstrap threshold at date $t$.}
\label{alg:bootstrap}
\hrule\smallskip
\noindent\textbf{Input:} return windows $W_{t-1},W_t$ (overlap $M-s$, stride $s$);
covariance estimator $\mathcal{A}$; monitoring dimension $K$; bootstrap draws $B$;
level $\alpha$ (default $0.05$).
\smallskip\hrule\smallskip
\begin{enumerate}\setlength{\itemsep}{2pt}
  \item $C_t\gets\mathcal{A}(S_t)$ from the current window sample covariance.
  \item \textbf{for} $b=1,\dots,B$:
  draw $r^{(b)}_1,\dots,r^{(b)}_{M+s}\overset{iid}{\sim}N(0,C_t)$; form the two
  windows $\{1,\dots,M\}$ and $\{s{+}1,\dots,s{+}M\}$; apply $\mathcal{A}$ to each
  window sample covariance and take top-$K$ projectors
  $\widehat P^{(b)}_{K,t-1},\widehat P^{(b)}_{K,t}$; set
  $\widehat D^{(b)}_{K,t}\gets\|\widehat P^{(b)}_{K,t}-\widehat P^{(b)}_{K,t-1}\|_F$.
  \item $q^{\mathrm{bt}}_{1-\alpha,t}\gets$ empirical $(1-\alpha)$ quantile of
  $\{\widehat D^{(b)}_{K,t}\}_{b=1}^{B}$.
\end{enumerate}
\smallskip\hrule\smallskip
\noindent\textbf{Output:} threshold $q^{\mathrm{bt}}_{1-\alpha,t}$; flag
$\mathbb 1\{\widehat D_{K,t}>q^{\mathrm{bt}}_{1-\alpha,t}\}$ for the observed movement
$\widehat D_{K,t}$.
\smallskip\hrule
\end{algorithm}

At fixed $N$ this bootstrap is consistent.
If the plug-in is consistent, $C_M\to_p\Sigma$ (the sample covariance, or QIS/LW
as $M\to\infty$), and the null cdf of $\widehat D_{K,t}$ is continuous and
strictly increasing at its $95\%$ quantile, then the simulated quantile
converges, $q_{95}(C_M)\to_p q_{95}(\Sigma)$, and the test
$\mathbb 1\{\widehat D_{K,t}>q_{95}(C_M)\}$ has asymptotic size $0.05$; the
argument is an almost-sure coupling with quantile convergence
(Appendix~\ref{app:bootstrap}), and the high-dimensional analogue is left
open.
Size and power of the band, the analytic test, and the bootstrap are compared in
the controlled two-window experiment of \S\ref{sec:tier-power}.

\subsection{Local power and the detectability frontier}
\label{sec:tier-b-power-theory}

The experiment of \S\ref{sec:tier-power} shows low power at onset for every
procedure. This reflects two-window monitoring itself rather than the particular
tests, and the following proposition makes the limitation precise, working from
the null simulator of Proposition~\ref{prop:noise-floor} alone.

\begin{proposition}[Local power under rotation alternatives]
\label{prop:tier-b-power}
Let windows A and B share $M-s$ observations, the $s$ new observations of
window B drawn i.i.d.\ from $\Sigma_2=R_\theta\Sigma_1R_\theta^\top$ with
$R_\theta$ a rotation of the plane $(u_1,u_b)$, $b>K$; let $\phi=s/M$ (post
design, disjoint windows, $s=M$, $\phi=1$) and
$D_{\mathrm{true}}=\|P_K(\Sigma_2)-P_K(\Sigma_1)\|_F=\sqrt2\,\sin\theta$.
\begin{enumerate}[(i)]
  \item \textbf{(Exact attenuation.)} $\E S_B=\bar\Sigma=(1-\phi)\Sigma_1+\phi\Sigma_2$
  exactly, so the estimable signal is
  $D_{\mathrm{mix}}(\phi,\theta)=\|P_K(\bar\Sigma)-P_K(\Sigma_1)\|_F$, with
  \[
    D_{\mathrm{mix}}(\phi,\theta)=\phi\cos\theta\,D_{\mathrm{true}}+O(\theta^2).
  \]
  At $\theta=\pi/2$ the mixture is diagonal in the population eigenbasis, so
  $D_{\mathrm{mix}}(\phi,\pi/2)=0$ whenever the mixed diagonal preserves the
  top-$K$ ordering; since the rotation alters only the diagonal entries $1$
  and $b$, the conditions $\phi<\tfrac12$,
  $\phi(\lambda_1-\lambda_b)<\lambda_K-\lambda_b$, and
  $\phi(\lambda_1-\lambda_b)<\lambda_1-\lambda_{K+1}$ are sufficient.
  Under these conditions, which hold at the experiment's onset fraction
  $\phi\approx0.25$, onset attenuation is nonlinear in $\theta$ and no
  two-window statistic can detect a quarter-turn rotation.
  Past the ordering threshold the picture reverses; in the post design
  ($\phi=1$) the mixture equals $\Sigma_2$, whose top-$K$ set at
  $\theta=\pi/2$ contains $u_b$ in place of $u_1$, and
  $D_{\mathrm{mix}}(1,\theta)=D_{\mathrm{true}}=\sqrt2\,\sin\theta$ at
  every angle.
  \item \textbf{(First-order law.)} If the returns are Gaussian and
  $\eta_s\vee\phi\|\Sigma_2-\Sigma_1\|_{\op}<\Delta_K/4$ for both windows, then
  to first order $\widehat D_{K,t}^2$ is distributed as the \emph{noncentral}
  version of Proposition~\ref{prop:noise-floor},
  $\sum_{i\le K<j}\tilde w_{ij}(Z_{ij}+\delta_{ij})^2$, where $\tilde w_{ij}$ and
  $Z_{ij}\sim N(0,1)$ are the null weights and standard variates implied by
  \eqref{eq:exact-null} and $\delta_{ij}=\E g_{ij}/\mathrm{sd}(g_{ij})$ is the
  standardized mean of the $(i,j)$ increment coordinate (nonzero only at $(1,b)$
  under the alternative), with total noncentrality
  \begin{equation}
    \textstyle\sum_{i\le K<j}\tilde w_{ij}\delta_{ij}^2
    =\phi^2\,\|L(\Sigma_2-\Sigma_1)\|_F^2
    =\tfrac{1}{2}\phi^2\sin^2(2\theta)
    \label{eq:ncp}
  \end{equation}
  which is \emph{spectrum-free}; the spectrum-dependent variance distortions of the
  $\Sigma_2$-blocks enter $\E\widehat D^2$ only at $O(\theta^2/M)$ for fixed
  $\phi$ (in simulation, at $M{=}1008$ the noncentrality matches \eqref{eq:ncp}
  within $15\%$ across spectra with opposite gap profiles).
  \item \textbf{(Frontier.)} Under the median-shift approximation
  $\mathrm{med}(\widehat D^2_\theta)\approx q_{50}^2+\mathrm{ncp}$, the
  50\%-power angle solves
  \begin{equation}
    \sin(2\theta_{50})
    =\frac{\sqrt{2\,(q_{1-\alpha}^2-q_{50}^2)}}{\phi},
    \label{eq:theta50}
  \end{equation}
  with both quantiles from the null simulator of
  Proposition~\ref{prop:noise-floor}.
  Because null quantiles of $\widehat D^2$ scale as $s/M^2$ under the
  weighted-$\chi^2$ approximation, for small angles
  $\theta_{50}\asymp C_\alpha(\lambda,K)/\sqrt{n_{\mathrm{post}}}$ where
  $n_{\mathrm{post}}$ is the number of post-break observations in window B
  ($s$ at onset, $M$ post). The frontier is governed by fresh observations
  alone, $\theta^{\mathrm{onset}}_{50}/\theta^{\mathrm{post}}_{50}=\sqrt{M/s}$,
  and the detection delay for a rotation of size $\theta$ scales as
  $1/\theta^2$.
  \item \textbf{(Onset power ceiling.)} If the right side of \eqref{eq:theta50}
  exceeds $1$, \emph{no} rotation angle attains $50\%$ power to first order, since the noncentrality \eqref{eq:ncp} is maximized at $\theta=\pi/4$ at the value
  $\phi^2/2$.
  Beyond first order the ceiling is governed by
  $\max_\theta D^2_{\mathrm{mix}}(\phi,\theta)$ via (i), attained past
  $\pi/4$; if $\max_\theta D^2_{\mathrm{mix}}<q^2_{1-\alpha}-q^2_{50}$ the
  sub-$50\%$ ceiling persists non-perturbatively.
\end{enumerate}
\end{proposition}

\noindent\textit{Proof.} Appendix~\ref{app:tier-b-power} (means and variances
of the exact increment under the alternative; the rank-one signal lies in the
$(1,b)$ pair coordinate; mixture eigenstructure for (i) and (iv)).

\noindent The smallest rotation a two-window monitor can detect shrinks only
as fast as fresh data accumulate, $\theta_{50}\asymp C/\sqrt{n_{\mathrm{post}}}$, so
halving the detectable angle requires four times as many post-break observations; and
at onset, where only a fraction $\phi$ of window~B is new, a quarter-turn rotation can
be undetectable at every angle.

\begin{corollary}[Closed-form detectability frontier]
\label{cor:theta-min}
Let $n_{\mathrm{post}}$ be the number of post-break observations in window B
($s$ at onset, $M$ post) and let $\tilde q_p$ denote the $p$-quantile of the
normalized null $Q=(M^2/4n_{\mathrm{post}})\,\widehat D^2$ (the fresh-data
count and the noise scale coincide in both designs).
The frontier \eqref{eq:theta50} unifies across designs as
\begin{equation}
  \sin^2(2\theta_{\min})
  =\frac{8\,(\tilde q_{1-\alpha}-\tilde q_{50})}{n_{\mathrm{post}}}.
  \label{eq:theta-min-unified}
\end{equation}
Under the weighted-$\chi^2$ null $Q=\sum_{i\le K<j}w_{ij}Z_{ij}^2$,
$w_{ij}=\lambda_i\lambda_j/(\lambda_i-\lambda_j)^2$, the two-moment normal
approximation $\tilde q_p\approx\sum w_{ij}+z_p(2\sum w_{ij}^2)^{1/2}$ has
vanishing median term ($z_{50}=0$), so the quantile gap collapses to
$z_{1-\alpha}(2\sum w_{ij}^2)^{1/2}$ and the frontier becomes fully explicit
in $(\lambda,K,n_{\mathrm{post}},\alpha)$,
\begin{equation}
  \sin(2\theta_{\min})
  =\bigl(8\sqrt2\,z_{1-\alpha}\bigr)^{1/2}
  \frac{\bigl(\sum_{i\le K<j}w_{ij}^2\bigr)^{1/4}}{\sqrt{n_{\mathrm{post}}}},
  \qquad
  n^{*}=8\sqrt2\,z_{1-\alpha}\Bigl(\textstyle\sum_{i\le K<j}
  w_{ij}^2\Bigr)^{1/2},
  \label{eq:theta-min-explicit}
\end{equation}
where $n^{*}$ is the minimum post-break sample at which \emph{any} rotation
angle is $50\%$-detectable (the ceiling condition of part~(iv) is
$n_{\mathrm{post}}<n^{*}$).
The eigengap enters through the weights, since every cross-cut pair has
$\lambda_i-\lambda_j\ge\Delta_K$, so $w_{ij}\le\lambda_i\lambda_j/\Delta_K^2$,
and shrinking gaps inflate $\theta_{\min}$ and $n^{*}$ at rate $\Delta_K^{-1}$
when the cut-adjacent pair dominates.
These expressions rest on the median-shift approximation of
Proposition~\ref{prop:tier-b-power}(iii) and on the two-moment normal
approximation of the null quantiles; we use them as a design rule rather than as
an exact power guarantee, and \S\ref{sec:tier-power} reports the gap between the
rule and the simulated power.
\end{corollary}

The three procedures suit different conditions. The Davis--Kahan band of
\S\ref{sec:theory} assumes no sampling model and holds for every estimator and
every distribution, but delivers only a gauge of whether the eigenspace is
identified, with no significance level. The analytic null of \S\ref{sec:tier-b} is
the cheapest calibrated test and is accurate at the ends of the overlap range,
near-disjoint or near-identical windows, under light tails, though its size drifts
at intermediate overlap and it assumes rotation-equivariant cleaning. The
estimator-aware bootstrap of \S\ref{sec:tier-c} covers the remaining cases,
intermediate overlap, clipping-based estimators, and heavy tails, at the cost of
repeated estimator fits per date.

\section{Calibrated inference for the scalar functionals}
\label{sec:functional-inference}

Besides the projector movement, the monitor reports two scalar summaries of the
spectrum, the absorption ratio $\AR_K$ and the leading-eigenvalue share $f_1$.
Both are scale-invariant ratio functionals of the eigenvalues. Shrinkage and
cleaning are designed to reduce matrix risk, which places no bound on the error in
such functionals, so their accuracy requires a separate analysis.

\subsection{Deterministic error bounds for the scalar functionals}
\label{sec:theory-functional}

Matrix loss $\|C-\Sigma\|_F$ ranks estimators by absolute eigenvalue error, whereas
the absorption ratio $\AR_K$ depends only on the normalized eigenvalues
$p_i=\lambda_i/\sum_j\lambda_j$ and is scale-invariant. The two orderings can
therefore disagree. An exact expansion of the ratio error in the eigenvalue
perturbations quantifies the discrepancy.

Let $e_i=\lambda_i(C)-\lambda_i(\Sigma)$,
$T=\sum_{j=1}^N\lambda_j(\Sigma)$,
$S_K=\sum_{i=1}^K\lambda_i(\Sigma)$,
$u_K=\sum_{i=1}^K e_i$, $v=\sum_{j=1}^N e_j$.

\begin{proposition}[Absorption-ratio expansion]
\label{prop:ar-expansion}
With $\AR_K(C)=S_K(C)/T(C)$ and $\AR_K(\Sigma)=S_K/T$ (both totals positive,
$T,\,T+v>0$),
\begin{equation}
  \AR_K(C)-\AR_K(\Sigma)
  =\frac{u_K}{T}-\frac{S_K v}{T^2}+R_K,
  \label{eq:ar-expand}
\end{equation}
where the remainder is exact,
\begin{equation}
  R_K=-\frac{(u_K T-S_K v)\,v}{T^2(T+v)}.
  \label{eq:ar-remainder}
\end{equation}
\end{proposition}

\begin{proof}
Subtracting the two ratios over the common denominator $T(T+v)$,
$\AR_K(C)-\AR_K(\Sigma)=(S_K+u_K)/(T+v)-S_K/T=(u_K T-S_K v)/[T(T+v)]$; the
first-order part $L=u_K/T-S_K v/T^2=(u_K T-S_K v)/T^2$ shares this numerator, so
\[
  R_K=\AR_K(C)-\AR_K(\Sigma)-L
  =(u_K T-S_K v)\Bigl(\frac{1}{T(T+v)}-\frac{1}{T^2}\Bigr)
  =-\frac{(u_K T-S_K v)\,v}{T^2(T+v)},
\]
which is \eqref{eq:ar-remainder}.
\end{proof}

In the first-order part of \eqref{eq:ar-expand}, $u_K/T$ is the error
accumulated inside the top block and $-S_Kv/T^2$ is the total error acting
through the normalization with leverage $S_K/T$. The two terms cancel exactly
for proportional errors $e_i=c\lambda_i$, for which $u_KT=S_Kv$, so a large but
proportional error leaves the ratio untouched while a small error split
unevenly across the cut does not.

\begin{proposition}[Matrix risk versus ratio error]
\label{prop:matrix-vs-ar}
If $\frac{1}{K}\sum_{i\le K}\lambda_i(\Sigma)\ne\frac{1}{N}\sum_j\lambda_j(\Sigma)$,
then for small $\varepsilon>0$, $C_A=\Sigma+\varepsilon I$ and $C_B=(1+b)\Sigma$ with
$b>\varepsilon\sqrt{N}/\|\Sigma\|_F$ satisfy
$\|C_A-\Sigma\|_F<\|C_B-\Sigma\|_F$ but
$|\AR_K(C_B)-\AR_K(\Sigma)|<|\AR_K(C_A)-\AR_K(\Sigma)|$.
\end{proposition}

\begin{proof}
The Frobenius distances are
$\|C_A-\Sigma\|_F=\|\varepsilon I\|_F=\varepsilon\sqrt{N}$ and
$\|C_B-\Sigma\|_F=\|b\Sigma\|_F=b\|\Sigma\|_F$, so the choice
$b>\varepsilon\sqrt{N}/\|\Sigma\|_F$ ensures $\|C_A-\Sigma\|_F<\|C_B-\Sigma\|_F$.
For $C_B$, scale invariance gives $\AR_K(C_B)=\AR_K(\Sigma)$ (numerator and
denominator of the ratio scale by $1+b$ and cancel), so
$|\AR_K(C_B)-\AR_K(\Sigma)|=0$.
For $C_A$, the eigenvalues are $\lambda_i(\Sigma)+\varepsilon$, hence
\[
  \AR_K(C_A)=\frac{S_K+K\varepsilon}{T+N\varepsilon},
\]
and we set $\phi(\varepsilon)=(S_K+K\varepsilon)/(T+N\varepsilon)$.
Differentiating at $\varepsilon=0$,
\[
  \phi'(\varepsilon)
  =\frac{K(T+N\varepsilon)-(S_K+K\varepsilon)N}{(T+N\varepsilon)^2}
  \;\Rightarrow\;
  \phi'(0)=\frac{KT-NS_K}{T^2}.
\]
The non-spherical condition $\frac{S_K}{K}\ne\frac{T}{N}$ is equivalent to
$KT-NS_K\ne 0$, so $\phi'(0)\ne 0$ and $\phi(\varepsilon)\ne S_K/T$ for small
$\varepsilon>0$. Thus $|\AR_K(C_A)-\AR_K(\Sigma)|>0$ while the $C_B$ error vanishes,
which completes the comparison.
\end{proof}

\begin{figure}[!ht]
  \centering
  \includegraphics[width=0.72\linewidth]{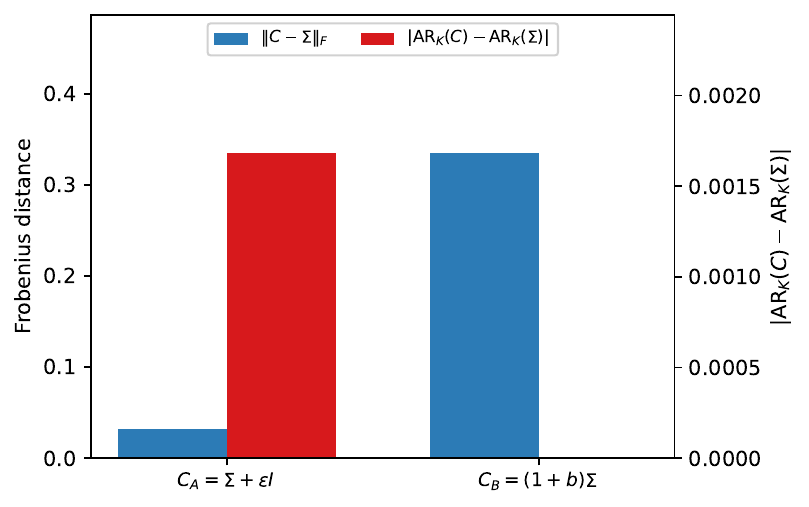}

  \caption{Diagonal example ($N{=}10$, $K{=}3$), where $C_A=\Sigma+\varepsilon I$ has
  smaller Frobenius distance than scaled $C_B=(1+b)\Sigma$ but larger $|\mathrm{AR}_K|$ error
  (Proposition~\ref{prop:matrix-vs-ar}).}
  \label{fig:matrix-vs-ar}
\end{figure}

The same separation holds for the other functionals. The proof of
Proposition~\ref{prop:matrix-vs-ar} uses only scale invariance and
$\phi'(0)\ne 0$, so it applies verbatim to $f_1=\lambda_1/T$ and to the spectral
entropy of $p_i=\lambda_i/\sum_j\lambda_j$.

\begin{proposition}[Functional error propagation]
\label{prop:propagate}
Fix positive constants $T_{\min},T_{\max},p_{\min}$ and a neighborhood $\mathcal U$
of the target spectrum on which $T_{\min}\le\sum_{j}\lambda_j\le T_{\max}$ and
$p_i=\lambda_i/\sum_j\lambda_j\ge p_{\min}$ for all $i$. For fixed $N$ and
$g\in\{f_1,\AR_K,H\}$ there is $L_{g,N}<\infty$ such that, whenever the spectra of
$C_t$ and $\Sigma_t$ lie in $\mathcal U$ with $\eta_t=\|C_t-\Sigma_t\|_{\op}$,
\[
  |g(C_t)-g(\Sigma_t)|\le L_{g,N}\,\eta_t.
\]
\end{proposition}

\begin{proof}
By Lemma~\ref{lem:weyl}, $|\lambda_i(C_t)-\lambda_i(\Sigma_t)|\le\eta_t$ for all $i$,
so $\|\lambda(C_t)-\lambda(\Sigma_t)\|_2\le\sqrt N\,\eta_t$. Each $g\in\{f_1,\AR_K,H\}$
is $C^1$ on $\mathcal U$ with bounded gradient (Appendix~\ref{app:propagate}), so the
mean-value theorem gives $|g(C_t)-g(\Sigma_t)|\le L_{g,N}\,\eta_t$ after absorbing
$\sqrt N$ into the constant.
\end{proof}

As in \S\ref{sec:theory}, all estimation error enters through the operator-norm
budget $\eta_t$. The constant $L_{g,N}$ is finite only on spectra bounded away
from degeneracy, and at $N/M$ of order one the plug-in bias of
\S\ref{sec:ar-hd-caution}, not the Lipschitz term, dominates the error.

\subsection{Scale invariance and elliptical calibration}
\label{sec:ar-inference}

Two facts established so far have a common cause. The elliptical null of the
analytic test scales by exactly $\sqrt{1+\kappa}$
(Remark~\ref{rem:elliptical-null}), and the common-kurtosis cross term cancels
exactly in the absorption-ratio expansion of \S\ref{sec:theory-functional}. Both
follow from first-order kurtosis immunity holding exactly for scale-invariant
functionals of the covariance, and only for them.
Throughout this subsection we assume the returns are i.i.d.\ \emph{elliptical} with
covariance $\Sigma$, finite fourth moments, and kurtosis parameter $\kappa$ (Gaussian
$\kappa=0$; marginal excess kurtosis $3\kappa$). Ellipticity fixes the law up to the
covariance, so the fourth-moment tensor depends on the family only through the scalar
$\kappa$ \citep{Tyler1981,Browne1984,Muirhead1982},
\begin{equation}
  M\,\mathrm{Cov}(S_{ij},S_{kl})\;\to\;
  (1+\kappa)\bigl(\Sigma_{ik}\Sigma_{jl}+\Sigma_{il}\Sigma_{jk}\bigr)
  +\kappa\,\Sigma_{ij}\Sigma_{kl}.
  \label{eq:ell-fourth-moment}
\end{equation}

The elliptical assumption does real work here. When it fails, for example under
skewed, non-elliptical returns, the $(1+\kappa)$ correction no longer calibrates the
intervals (\S\ref{sec:mc-si-coverage}, Table~\ref{tab:non-elliptical}), and the
estimator-aware bootstrap or a distribution-free fallback is required.

\begin{proposition}[Kurtosis immunity and scale invariance]
\label{prop:scale-immunity}
Fix $N$ and let $M\to\infty$ for i.i.d.\ elliptical returns obeying
\eqref{eq:ell-fourth-moment}, and let
$W:\mathcal{S}\to\R$ be defined and continuously differentiable on an open subset
$\mathcal{S}$ of the positive-definite cone, with symmetric gradient
$G=\nabla W(\Sigma)$.
\begin{enumerate}[(i)]
  \item \textbf{(Variance decomposition.)}
  $\sqrt{M}\bigl(W(S_M)-W(\Sigma)\bigr)\to_d N\bigl(0,\sigma^2_W(\Sigma,\kappa)\bigr)$ with
  \begin{equation}
    \sigma^2_W(\Sigma,\kappa)
    =2(1+\kappa)\,\mathrm{tr}\bigl[(G\Sigma)^2\bigr]
    +\kappa\,\langle G,\Sigma\rangle^2,
    \qquad \langle G,\Sigma\rangle=\mathrm{tr}(G\Sigma).
    \label{eq:si-variance}
  \end{equation}
  \item \textbf{(Immunity.)} If $W$ is scale-invariant near $\Sigma$
  ($W(b\Sigma)=W(\Sigma)$ for $b$ near $1$), then $\langle G,\Sigma\rangle=0$
  (Euler's relation for degree-zero homogeneity) and
  \[
    \sigma^2_W(\Sigma,\kappa)=(1+\kappa)\,\sigma^2_W(\Sigma,0),
  \]
  so ellipticity enters only through the scalar factor $(1+\kappa)$ and the Gaussian
  Wald interval rescaled by $\sqrt{1+\widehat\kappa}$ is asymptotically valid across
  the entire elliptical class.
  \item \textbf{(Converse.)} Suppose
  $\sigma^2_W(\Sigma,\kappa_0)=(1+\kappa_0)\,\sigma^2_W(\Sigma,0)$ for a single
  $\kappa_0>0$ at every $\Sigma$ in a ray-connected open set
  $\mathcal{C}\subseteq\mathcal{S}$ (i.e.\ $\{b>0:b\Sigma\in\mathcal{C}\}$ is an
  interval for each $\Sigma\in\mathcal{C}$).
  Then $\langle\nabla W(\Sigma),\Sigma\rangle=0$ on $\mathcal{C}$ and $W$ is constant
  along every ray segment in $\mathcal{C}$. First-order kurtosis immunity holds
  \emph{only} for scale-invariant functionals.
\end{enumerate}
\end{proposition}

\noindent\textit{Proof.} Appendix~\ref{app:scale-immunity} (delta-method variance as
the variance of a scalar quadratic form via \eqref{eq:ell-fourth-moment}; Euler's
relation in both directions).
The common-shock term $\kappa\langle G,\Sigma\rangle^2$ is the price of measuring
\emph{level}. For $W=\mathrm{tr}$ it equals $\kappa T^2$ and dominates the variance
under heavy tails, whereas every ratio functional sheds it identically.

\begin{lemma}[Euler identities for the monitored functionals]
\label{lem:ar-orth}
The functionals $\AR_K$, $f_1$, and $H$ are scale-invariant with matrix gradients
(at $\Sigma$ with $\lambda_K>\lambda_{K+1}$ for $\AR_K$, $\lambda_1$ simple for
$f_1$, $\lambda_i>0$ for $H$; $P_K$ the top-$K$ projector, $p_i=\lambda_i/T$)
\[
  \nabla\AR_K=\frac{P_K}{T}-\frac{S_K}{T^2}I,\qquad
  \nabla f_1=\frac{u_1u_1^\top}{T}-\frac{\lambda_1}{T^2}I,\qquad
  \nabla H=-\frac{1}{T}\sum_i\bigl(\log p_i+H\bigr)u_iu_i^\top,
\]
each satisfying $\langle\nabla W,\Sigma\rangle=0$ exactly.
In eigenvalue coordinates, $g_i=\partial\AR_K/\partial\lambda_i
=(\mathbf 1_{\{i\le K\}}T-S_K)/T^2$ and $\sum_i g_i\lambda_i=0$, as used in
Proposition~\ref{prop:ar-clt}.
Note $\AR_K$ requires only the gap condition $\lambda_K>\lambda_{K+1}$, not simple
eigenvalues, since ties inside the top block leave $\mathrm{tr}(P_K\Sigma)$ smooth.
\end{lemma}

\begin{corollary}[Elliptical projector null]
\label{cor:dhat-ell}
Under i.i.d.\ elliptical sampling, the top-$K$ projector is scale-invariant, and in the
population eigenbasis the common-shock term of \eqref{eq:ell-fourth-moment} is
$\kappa\lambda_i\lambda_j\delta_{ij}\delta_{kl}$, which vanishes on the off-diagonal
entries $E_{ij}$, $i\le K<j$, which are the only entries appearing in the linear term
$L(\Delta E)$ of Proposition~\ref{prop:noise-floor}.
Hence under elliptical sampling every pair variance in the first-order null scales
by exactly $(1+\kappa)$ and the null quantiles of $\widehat D_{K,t}$ scale by
$\sqrt{1+\kappa}$, so the elliptical correction of Remark~\ref{rem:elliptical-null} is the
projector instance of Proposition~\ref{prop:scale-immunity}(ii), not a separate
approximation.
\end{corollary}

\begin{proposition}[Exact window-scale pivotality]
\label{prop:exact-pivot}
Let $W$ be scale-invariant and let the window returns satisfy $r_\tau=c\,z_\tau$
for $\tau=1,\ldots,M$, where $c>0$ is a single random volatility scale common to
the whole window and $z_\tau\in\R^N$ are the underlying unit-scale returns
(arbitrary dependence between $c$ and $\{z_\tau\}$ is allowed).
Then $S_r=c^2S_z$ and $W(S_r)=W(S_z)$ \emph{for every realization}, so the law of
$W(S)$ does not depend on the distribution of $c$ at all, and the statement is
exact in finite samples, not asymptotic.
Moreover, if $\mathcal{A}$ is scale-equivariant ($\mathcal{A}(c^2S)=c^2\mathcal{A}(S)$,
true of QIS and linear shrinkage), then $\widehat D_{K,t}$, $\widehat\Delta_{K,t}$,
$\eta_t$, and hence the entire capped band $\tau^*_{K,t}$ and its flags are
invariant under a common scaling of both windows, so the monitor is
volatility-level-blind by construction and reacts only to shape.
\end{proposition}

\noindent\textit{Proof.} Appendix~\ref{app:exact-pivot}.
This exactness explains two empirical regularities, the procedures' silence under the
pure scale-change design of \S\ref{sec:tier-power} (finding~v) and the robustness
of absorption-ratio readings to volatility \emph{regimes} as opposed to correlation
changes.
Day-varying scales ($c_\tau$ changing within the window, i.e.\ volatility
clustering) break the exact invariance; that residual is what the moment-based
$\widehat\kappa$ absorbs on panels (\S\ref{sec:panel-inference}).

The immunity direction (ii) has classical roots.
\citet{Tyler1983} showed that
scale-invariant functions of scatter estimators have asymptotic distributions
depending on the elliptical family only through a single scalar, and
\citet{ShapiroBrowne1987} proved the corresponding robustness of normal-theory
methods for scale-invariant covariance-structure models (the kurtosis-corrected
sphericity test of \citealt{Muirhead1982} is an early instance); in the
shape-matrix literature, \citet{Paindaveine2008} showed that the
determinant-normalized shape component is the canonical one for which shape
and scale inference decouple asymptotically in elliptical families.
For the weakly identified case, \citet{PaindaveinePeralvoVerdebout2026} show
that multivariate rank tests for a hypothesized leading eigenvector of an
elliptical scatter matrix retain size and efficiency as the two leading
eigenvalues coalesce, the limit in which the band of \S\ref{sec:theory-prop}
saturates; their setting is a fixed direction tested from a single sample,
whereas the calibration here concerns movement of an estimated subspace
between overlapping windows.
We add the converse characterization (iii), by which immunity holds \emph{only}
on the scale-invariant class, the explicit instantiation for the spectral functionals
used in risk monitoring, the transfer to sequential projector movement
(Corollary~\ref{cor:dhat-ell}), and the exact finite-sample pivotality of
Proposition~\ref{prop:exact-pivot}.

\subsection{Confidence intervals for the absorption ratio and leading share}
\label{sec:ar-ci}

The immunity of Proposition~\ref{prop:scale-immunity} yields a confidence interval
for the absorption ratio, and for the leading share, whose only elliptical input
is the scalar $\kappa$.
Applying Proposition~\ref{prop:scale-immunity}(ii) with the gradients of
Lemma~\ref{lem:ar-orth} gives the limiting variances.

\begin{proposition}[Elliptical CLT for scale-invariant functionals]
\label{prop:ar-clt}
Fix $N$ and let $M\to\infty$ for i.i.d.\ elliptical returns obeying
\eqref{eq:ell-fourth-moment}.
Then by Proposition~\ref{prop:scale-immunity}(ii) and Lemma~\ref{lem:ar-orth},
\[
  \sqrt{M}\bigl(\widehat{\AR}_K-\AR_K\bigr)\;\to_d\;
  N\Bigl(0,\;\frac{2(1+\kappa)}{T^4}
  \bigl[(T-S_K)^2\textstyle\sum_{i\le K}\lambda_i^2+S_K^2\sum_{i>K}\lambda_i^2\bigr]\Bigr),
\]
and likewise
\[
  \sigma^2_{f_1}=\frac{2(1+\kappa)}{T^2}
  \Bigl[(1-f_1)^2\lambda_1^2+f_1^2\textstyle\sum_{i>1}\lambda_i^2\Bigr],
  \qquad
  \sigma^2_{H}=\frac{2(1+\kappa)}{T^2}
  \textstyle\sum_i\lambda_i^2\bigl(\log p_i+H\bigr)^2.
\]
The $\kappa\lambda_i\lambda_j$ cross-term vanishes identically for all three, so
scale-invariant functionals are first-order immune to common volatility
shocks, which is the mechanism that makes Wald intervals usable on heavy-tailed
returns.
\end{proposition}

In the $\AR_K$ variance the top-block eigenvalues enter with weight $(T-S_K)^2$
and the bulk eigenvalues with weight $S_K^2$, the squared leverages of the two
error terms in the expansion of \S\ref{sec:theory-functional}, and the variance
vanishes as $\AR_K$ approaches $0$ or $1$. The interval uses the plug-in
spectrum and the single estimated scalar $\widehat\kappa$; the behavior of that
estimate under heavy tails is the practical constraint.

\begin{remark}[Estimating $\kappa$ by radial MLE]
\label{rem:kappa-mle}
Moment-based kurtosis estimates are biased \emph{downward} for heavy tails
(anti-conservative).
We instead fit a multivariate-$t$ tail index $\nu$ by maximum likelihood on the
cross-sectional radius $\|r_t\|^2$, with a Satterthwaite effective-dimension correction
$d_{\mathrm{eff}}=(\mathrm{tr}\,S)^2/\mathrm{tr}(S^2)$ for spiked spectra, and set
$\widehat\kappa=2/(\widehat\nu-4)$ (capped; $0$ for $\widehat\nu>100$).
In MC at $N{=}115$, $M{=}252$ this is nearly unbiased at $t_{12}$ ($0.24$ vs $0.25$),
mildly conservative at $t_8$ ($0.64$ vs $0.50$), and \emph{conservative} for
extreme tails ($t_5$), the safe direction; with it, $t_5$ coverage of the fixed-$N$
interval reaches $90\%$ at $N{=}30$ and $88\%$ at $N{=}115$
(Table~\ref{tab:ar-coverage}).
\end{remark}

\begin{proof}[Proof of Proposition~\ref{prop:ar-clt}]
Immediate from
Proposition~\ref{prop:scale-immunity}(ii) with the gradients of
Lemma~\ref{lem:ar-orth}; the variance algebra is in
Appendix~\ref{app:ar-orth}.
\end{proof}

In practice the interval for a scale-invariant functional
$g\in\{\AR_K,f_1\}$ is $g(C)\pm z_{1-\alpha/2}\,\widehat\sigma_g/\sqrt{M-1}$, where
$z_{1-\alpha/2}$ is the standard-normal quantile and $\widehat\sigma_g^2$ is the
plug-in of the asymptotic variance stated in Proposition~\ref{prop:ar-clt}; its only
elliptical input is the single estimated scalar $\widehat\kappa$.

\subsection{The absorption ratio in high dimensions}
\label{sec:ar-hd-caution}

In high dimensions the absorption-ratio interval needs more than the elliptical
correction, and cleaning the covariance does not supply it.
At $N/M\approx 0.5$ the plug-in $\widehat{\AR}_K$ from sample eigenvalues is biased
upward by Marchenko--Pastur spreading, and plugging in QIS-cleaned eigenvalues
makes coverage \emph{worse}, $26\%$ in the Gaussian simulation, because
Frobenius-optimal shrinkage intentionally over-shrinks spikes relative to
population values, which is optimal for matrix loss but biased for spectral
functionals.
The appropriate correction inverts the spiked-model bias map
$\widehat\lambda=\lambda\bigl(1+c\,\sigma^2/(\lambda-\sigma^2)\bigr)$
\citep{Paul2007,BaikSilverstein2006} for eigenvalues above the bulk edge,
returns the spike inflation to the bulk so that the debiased spectrum is
trace-preserving (Proposition~\ref{prop:wedge}(iv)), and applies a chain-rule
adjustment to the delta-method variance.
Table~\ref{tab:ar-coverage} shows coverage at $N{=}115$, $M{=}252$ improving from $90\%$
(raw) to $96\%$ (debiased) for Gaussian data, with the mean bias falling from
$+0.0076$ to $+0.0003$.
Under elliptical heavy tails in high dimensions the MP edge itself is distorted; with
the radial-MLE kurtosis estimate the debiased intervals recover $91$--$92\%$ coverage
at $t_5$ (from $46\%$ with the Gaussian formula and $61\%$ with moment-based
$\widehat\kappa$ at $N{=}115$), still short of nominal; fully calibrated AR
inference for elliptical high-dimensional data remains an open problem.
Our population target also differs from the sample-principal-component estimand
of \citet{PerryPanigrahiBienWitten2025}, whose inference conditions on the
observed singular vectors.

The sample (Marchenko--Pastur) and cleaning biases of $\widehat{\AR}_K$ admit a
sharp asymptotic description in Proposition~\ref{prop:wedge}. In the spiked model
the plug-in absorption ratio from the sample covariance and
from \emph{any} Frobenius-optimal rotation-equivariant estimator both converge
to wrong values, with explicit and oppositely signed wedges of the same leading
order, while only spike debiasing is consistent.

\begin{proposition}[Cleaning--debiasing wedge]
\label{prop:wedge}
Let $\Sigma_N=\sigma^2I_N+\sum_{i\le K'}(\lambda_i-\sigma^2)v_iv_i^\top$ with
fixed distinct spikes
$\lambda_1>\cdots>\lambda_{K'}>\sigma^2(1+\sqrt c)$, $K\le K'$ fixed, and let
$N/M\to c\in(0,1)$ with Gaussian observations.
Write $S_K=\sum_{i\le K}\lambda_i$,
$\psi(\lambda)=\lambda\bigl(1+c\sigma^2/(\lambda-\sigma^2)\bigr)$, and let
$\alpha_i^2=\bigl(1-c/\rho_i^2\bigr)/\bigl(1+c/\rho_i\bigr)$ with
$\rho_i=(\lambda_i-\sigma^2)/\sigma^2$ denote the limiting sample--population
eigenvector overlap \citep{Paul2007,BenaychGeorgesNadakuditi2011}.
\begin{enumerate}[(i)]
  \item \textbf{(Sample, upward wedge.)}
  $\widehat{\AR}_K(S)/\AR_K(\Sigma_N)\to_{a.s.}1+\Delta_S$ with
  \[
    \Delta_S=\frac{c\,\sigma^2}{S_K}\sum_{i\le K}\frac{\lambda_i}{\lambda_i-\sigma^2}>0 .
  \]
  \item \textbf{(Frobenius oracle, downward wedge.)} The finite-sample
  Frobenius-optimal rotation-equivariant estimator
  $D^*=\sum_i(u_i^\top\Sigma u_i)\,u_iu_i^\top$ \citep{LedoitPeche2011} is
  exactly trace-preserving, its spike values converge to
  $\lambda_i\alpha_i^2+\sigma^2(1-\alpha_i^2)$, and
  \[
    \widehat{\AR}_K(D^*)/\AR_K(\Sigma_N)\to_{a.s.}1-\Delta_O,
    \qquad
    \Delta_O=\frac{c\,\sigma^2}{S_K}\sum_{i\le K}
    \frac{\lambda_i}{\lambda_i-\sigma^2+c\,\sigma^2}>0,
  \]
  using the identity
  $(\lambda-\sigma^2)(1-\alpha^2)=c\sigma^2\lambda/(\lambda-\sigma^2+c\sigma^2)$;
  hence $0<\Delta_O<\Delta_S$ term by term.
  \item \textbf{(Debiasing works.)} The trace-preserving spike-debiased plug-in
  \[
    \widehat{\AR}{}^{\mathrm{deb}}_K
    =\frac{\sum_{i\le K}\psi^{-1}(\widehat\lambda_i)}{\mathrm{tr}\,S}
  \]
  ($\psi^{-1}$ applied above the edge, with the spike inflation returned to the
  bulk so that the debiased spectrum sums exactly to $\mathrm{tr}\,S$) is
  strongly consistent for $\AR_K(\Sigma_N)$.
  Because $\AR_K(\Sigma_N)=S_K/\mathrm{tr}\,\Sigma_N=O(1/N)$ in this regime, the
  central limit theorem is stated for the \emph{relative} error. The numerator
  $\sum_{i\le K}\psi^{-1}(\widehat\lambda_i)$ is $\sqrt M$-asymptotically normal
  around $S_K$, the denominator $\mathrm{tr}\,S/\mathrm{tr}\,\Sigma_N\to_{a.s.}1$
  with fluctuation negligible at the $\sqrt M$ scale, and hence
  \[
    \sqrt M\Bigl(\widehat{\AR}{}^{\mathrm{deb}}_K/\AR_K(\Sigma_N)-1\Bigr)
    \;\to_d\;N\bigl(0,\,V_{\mathrm{rel}}\bigr),
  \]
  with $V_{\mathrm{rel}}$ of order one (the wedges $\Delta_S,\Delta_O$ above are
  the corresponding deterministic relative biases).
\end{enumerate}
\end{proposition}

\noindent\textit{Proof.} Appendix~\ref{app:wedge}, together with the
strong-spike limit of the two wedges, the finite-$N$ cost of dropping trace
preservation, and the underlying zero-wedge characterization. Within
trace-preserving rotation-equivariant estimators with a deterministic spike
map, the wedge vanishes only when the map is the identity above the bulk edge,
so Frobenius-optimal cleaning is never functional debiasing.
The practical implication for the monitor is direct.
The absorption ratio must be computed from spike-debiased eigenvalues;
computing it from the cleaned matrix that drives the projector replaces one
known bias by another of the opposite sign.

\citet{DonohoGavishJohnstone2018} established that the optimal eigenvalue
shrinker depends on the choice of matrix loss; Proposition~\ref{prop:wedge} is
a scalar-functional counterpart for the ratio statistic that risk monitors
report.
The wedge formulas are quantitatively sharp at the scale of our designs and are
verified across concentrations in \S\ref{sec:mc-si-coverage}; estimation of
$\sigma^2$ and $c$ for the debiasing map follows \citet{PassemierLiYao2017}.

\section{Simulation studies}
\label{sec:mc}

The simulations use factor designs with a known population covariance $\Sigma$, so
every estimate is compared against the truth rather than a proxy.
Returns are generated from factor models at $N{=}115$ unless stated otherwise,
with baseline replication counts of $10^3$ at $M{=}252$ and $500$ replications
in robustness cells, under fixed seeds.

\subsection{Coverage of the perturbation band}
\label{sec:mc-coverage}

The Davis--Kahan bound depends on raw gaps $\Delta_{K,t}$; we contrast weak-gap
and strong-gap factor designs (loadings scaled toward $\lambda_1\approx\lambda_2$
or separated eigenvalues).
With $\eta_t=\|C_t-\Sigma_t\|_{\op}$ and population $\Delta_{K,t}$ in the capped band
$\tau^*_{K,t}$, Table~\ref{tab:mc-monitor} reports that
$|\widehat D_{K,t}-D_{K,t}|\le\tau^*_{K,t}$ on 100\% of simulated dates, together with
\emph{tightness ratios} $|\widehat D_{K,t}-D_{K,t}|/\tau^*_{K,t}$ (median and 95th
percentile) so that coverage is interpretable.
The uncapped band has a mean exceeding the maximum possible movement $\sqrt{2K}$ by
one to two orders of magnitude in these designs; the capped band is on the scale of
the monitored statistic.
At this $N/M$ the per-date Davis--Kahan terms saturate their caps on most dates
(weak identification), which is precisely the regime message of
Remark~\ref{rem:cdk}.

\begin{table}[!ht]
  \centering
  \caption{Simulated projector monitoring with the capped band $\tau^*_{K,t}$
  (population $\Delta_{K,t}$; ``Mean uncapped'' is the legacy
  $2\sqrt{2K}(\eta_t/\Delta_{K,t}+\eta_{t-1}/\Delta_{K,t-1})$, vacuous relative to the
  maximum movement $\sqrt{2K}{=}2$; tightness $=|\widehat D-D|/\tau^*$).}
  \label{tab:mc-monitor}
  \footnotesize
  \resizebox{\textwidth}{!}{\begin{tabular}{lrrrrrrr}
\toprule
Gap & Phase & Mean $\widehat D_{K,t}$ & Mean $\tau^*_{K,t}$ & Mean uncapped &
Med.\ / p95 tightness & Share bound holds & Share $\widehat D_{K,t}>\tau^*_{K,t}$ \\
\midrule
Strong & Calm & 0.64 & 2.00 & 92.7 & 0.27 / 0.64 & 100.0\% & 0.0\% \\
Strong & Break & 0.64 & 2.00 & 60.5 & 0.28 / 0.61 & 100.0\% & 0.0\% \\
Weak & Calm & 0.73 & 2.00 & 52.1 & 0.33 / 0.66 & 100.0\% & 0.0\% \\
Weak & Break & 0.70 & 2.00 & 125.1 & 0.31 / 0.64 & 100.0\% & 0.0\% \\
\bottomrule
\end{tabular}
}
\end{table}

\subsection{Calibration of $\alpha\tau^*_{K,t}$}
\label{sec:mc-alpha}

Table~\ref{tab:mc-alpha} and Figure~\ref{fig:mc-monitor-alpha} report
$\mathbb{P}(\widehat D_{K,t}>\alpha\tau^*_{K,t})$ by gap regime.
The flag rate is decreasing in $\alpha$; the operational scale $\alpha^\star$ is the
smallest grid value whose simulated calm-period rate stays $\le 5\%$
(Remark~\ref{rem:alpha-scaled-tau}), here $\alpha^\star{=}0.75$.
These rates document conservativeness of the scaled band, not power to detect latent
population movement.

\begin{table}[!ht]
  \centering
  \caption{Simulated flag rates $\mathbb{P}(\widehat D_{K,t}>\alpha\tau^*_{K,t})$ by gap regime.}
  \label{tab:mc-alpha}
  \footnotesize
  \begin{tabular}{lrrrrrr}
\toprule
Gap & Phase & $\alpha{=}0.25$ & $0.50$ & $0.75$ & $0.90$ & $1.00$ \\
\midrule
Strong & Calm & 62.5\% & 11.4\% & 0.0\% & 0.0\% & 0.0\% \\
Strong & Break & 65.4\% & 11.4\% & 0.0\% & 0.0\% & 0.0\% \\
Weak & Calm & 83.4\% & 15.9\% & 0.0\% & 0.0\% & 0.0\% \\
Weak & Break & 82.4\% & 13.1\% & 0.0\% & 0.0\% & 0.0\% \\
\bottomrule
\end{tabular}

\end{table}

\begin{figure}[!ht]
  \centering
  \includegraphics[width=0.72\linewidth]{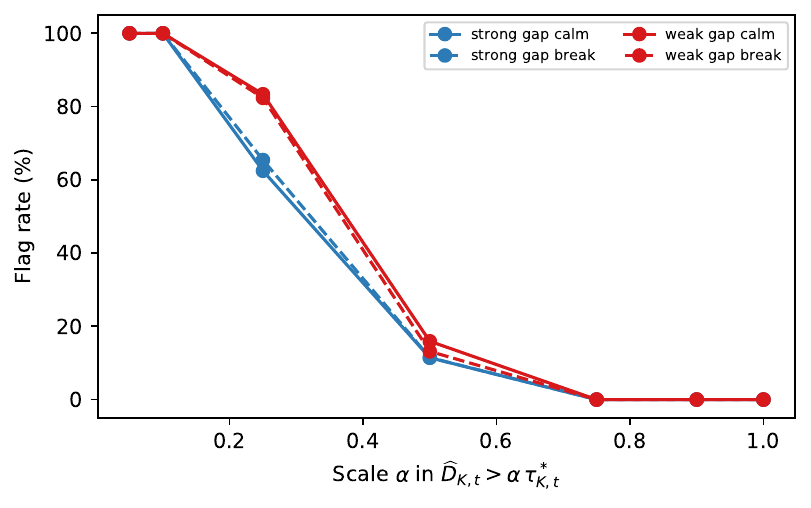}
  \caption{Simulated flag rate vs.\ $\alpha$ (known $\Sigma_t$, population $\tau^*_{K,t}$).}
  \label{fig:mc-monitor-alpha}
\end{figure}

\subsection{Functional-risk separation}
\label{sec:mc-functional}

We simulate $K_f=5$ factors, $N=115$, window $M{=}252$, and
$C_t\in\{S_t,\mathrm{QIS},\mathrm{LW},\text{MP-clip}\}$ (sample $S_t$ as baseline).
Table~\ref{tab:mc} reports mean population functional errors ($10^3$ replications);
Figure~\ref{fig:mc-functional} plots the same run.

\begin{table}[!ht]
  \centering
  \caption{Simulated mean error of the spectral functionals ($\Sigma$ known,
  $N{=}115$, $M{=}252$, $10^3$ reps; seed 42). $C=\mathcal{A}(S)$ is the cleaned
  (shrinkage) estimate; $\|\cdot\|_{\op}$ is the operator-norm matrix error;
  $\AR_K$ is the top-$K$ absorption ratio (variance share of the top $K$ eigenvalues)
  and $f_1$ the leading-eigenvalue share.}
  \label{tab:mc}
  \footnotesize
  \begin{tabular}{lrrr}
\toprule
Estimator & $\|C-\Sigma\|_{\op}$ & $|\AR_K$ err$|$ & $|f_1$ err$|$ \\
\midrule
Sample $S$ & $1.27e-04$ & 0.0877 & 0.0078 \\
QIS & $8.82e-05$ & 0.0155 & 0.0048 \\
LW & $1.03e-04$ & 0.0326 & 0.0115 \\
MP-clip & $9.85e-05$ & 0.0080 & 0.0009 \\
\bottomrule
\end{tabular}

  \par\vspace{0.5em}
  \parbox{\linewidth}{\footnotesize Sample $S_t$ has largest $\|C-\Sigma\|_{\op}$ and $|\AR_K|$ error.
  95\% MC CIs ($\pm 1.96\times$SE, $10^3$ reps) are QIS
  $\|C-\Sigma\|_{\op}$ $[8.76,8.88]\times 10^{-5}$;
  $|\AR_K|$ $[0.0153,0.0158]$;
  $|f_1|$ $[0.0047,0.0049]$.
  QIS beats LW on $\|\cdot\|_{\op}$ and $\AR_K$ in 100\% of $10^3$ replications.
  MP-clip attains the lowest mean $\AR_K$ and $f_1$ errors in this run, while QIS dominates LW
  on operator norm and $\AR_K$, illustrating Propositions~\ref{prop:ar-expansion}
  and~\ref{prop:matrix-vs-ar}, by which matrix-risk and ratio-functional rankings need not agree.}%
\end{table}

Figure~\ref{fig:mc-functional} visualizes the same means, with QIS improving on LW for
operator norm and $|\AR_K|$ error, while MP-clip attains the lowest $|\AR_K|$ and $|f_1|$
means in this DGP, a concrete instance of functional-risk separation.
\begin{figure}[!ht]
  \centering
  \includegraphics[width=\linewidth]{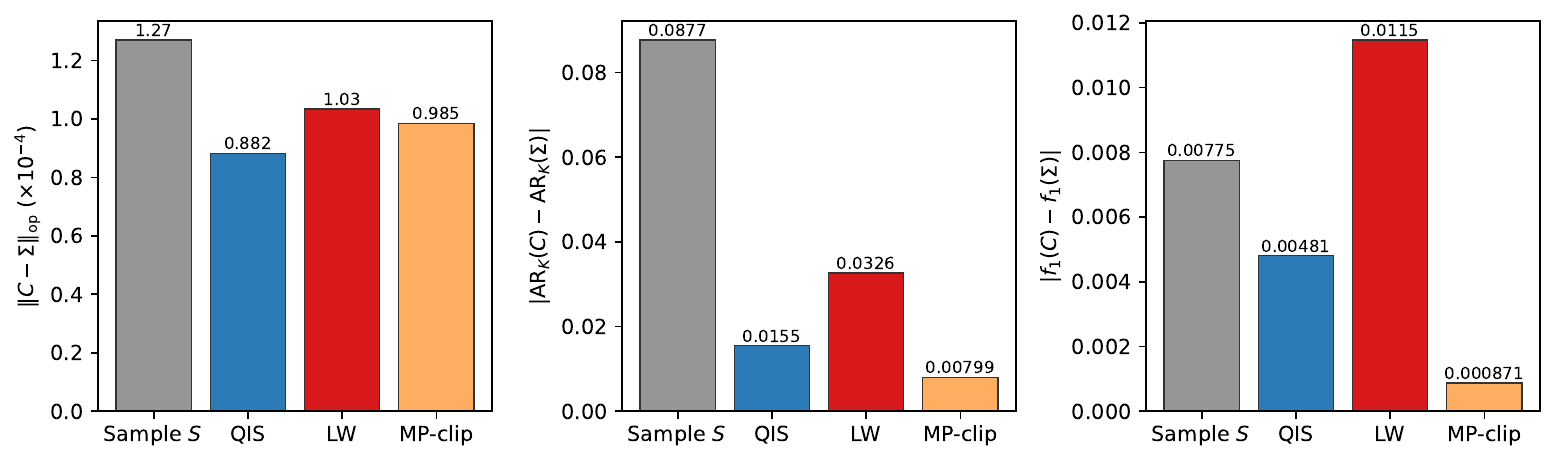}
  \caption{Simulated mean functional error by estimator (known $\Sigma$, $N{=}115$,
  $M{=}252$, $10^3$ replications; same run as Table~\ref{tab:mc}). Operator norm panel scaled
  by $10^4$ for readability.}
  \label{fig:mc-functional}
\end{figure}

\paragraph{Sensitivity to the shrinkage estimator.}
Table~\ref{tab:shrinkage-sensitivity} repeats the functional-error comparison across
shrinkage \emph{intensities}, namely sample covariance, linear shrinkage (LW),
oracle-approximating shrinkage (OAS, \citealp{ChenWieselHeroEldar2010}), nonlinear
QIS, and MP-clip. The two linear estimators are essentially indistinguishable (LW and
OAS agree to within Monte Carlo error on every metric), and the qualitative conclusion
is stable across intensity, with nonlinear QIS minimizing operator-norm error while MP-clip
minimizes the ratio-functional errors, so the matrix-risk and functional-risk rankings
separate regardless of the shrinkage map
(Propositions~\ref{prop:ar-expansion} and~\ref{prop:matrix-vs-ar}).

\begin{table}[!ht]
  \centering
  \caption{Sensitivity to the shrinkage estimator, simulated mean functional error
  ($\Sigma$ known, $N{=}115$, $M{=}252$, $5$ factors, $K{=}5$, $10^3$ reps; seed 42).
  $C=\mathcal{A}(S)$ is the cleaned estimate, $\|\cdot\|_{\op}$ the operator norm,
  $\AR_K$ the top-$K$ absorption ratio, and $f_1$ the leading-eigenvalue share.
  LW and OAS coincide to MC error; the operator-norm vs.\ ratio-functional ranking
  separation holds for every estimator.}
  \label{tab:shrinkage-sensitivity}
  \footnotesize
  \begin{tabular}{lrrr}
\toprule
Estimator & $\|C-\Sigma\|_{\op}$ & $|\AR_K(C)-\AR_K(\Sigma)|$ & $|f_1(C)-f_1(\Sigma)|$ \\
\midrule
Sample $S$ & 1.270e-04 & 0.0302 & 0.0078 \\
LW (linear) & 1.034e-04 & 0.0450 & 0.0115 \\
OAS & 1.038e-04 & 0.0453 & 0.0116 \\
QIS (nonlinear) & 8.818e-05 & 0.0259 & 0.0048 \\
MP-clip & 9.845e-05 & 0.0084 & 0.0009 \\
\bottomrule
\end{tabular}

\end{table}

\paragraph{Sensitivity to $N/M$.}
Varying the window so that $N/M\in\{0.25,0.5,1.0,1.5\}$ at $N{=}115$, QIS wins on
$\|\cdot\|_{\op}$ and $|\AR_K|$ error in every replication at all $M$ tested,
across both weak- and strong-gap designs.
On simulated paths stratified by $\gaptwelve$, Figure~\ref{fig:mc-gap} shows the
standard deviation of the leading-vector angle increment falling from $0.092$
(low-gap tertile) to $0.003$ (high-gap), while that of the subspace-distance
increment stays near $0.32$--$0.36$, so the subspace distance tracks the
$K$-dimensional state when the leading vector is poorly identified.

\begin{figure}[!ht]
  \centering
  \includegraphics[width=0.78\linewidth]{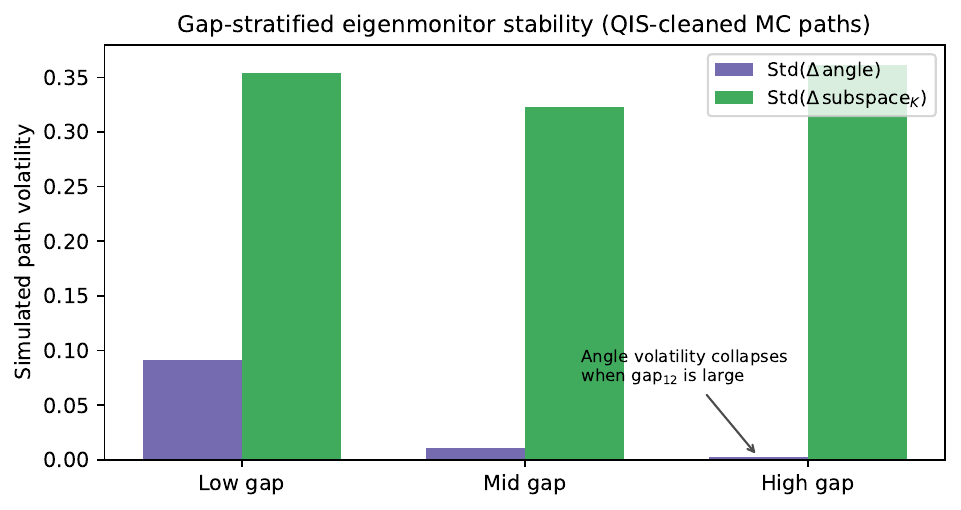}
  \caption{Simulated eigenmonitor volatility by $\gaptwelve$ tertile (known $\Sigma$ paths).}
  \label{fig:mc-gap}
\end{figure}

\subsection{Size and power of the calibrated tests}
\label{sec:tier-power}

Table~\ref{tab:tier-power} reports the controlled two-window experiment with
$N{=}50$, $M{=}126$, $K{=}2$, and $s{=}31$.
Window A is drawn from $\Sigma_1$, and window B either straddles the break
(onset) or lies entirely past it (post),
with $\Sigma_2$ a rotation of the top eigenvector by $\theta$; a scale-change
design ($\Sigma_2=1.5\Sigma_1$, no subspace movement) and elliptical-$t_5$
innovations complete the grid.
Under Gaussian innovations every procedure controls size.
The capped worst-case band has essentially no power, as expected of a bound on
the worst case, while the calibrated tests reach $93$--$95\%$ power at
$\theta{=}1$ in the post design.
Onset power is low for every method because the new regime occupies only $s$ of
the $M$ observations; this is the intrinsic detection delay of rolling-window
monitoring at this $N/M$, the detectability frontier that
Proposition~\ref{prop:tier-b-power} derives in closed form together with the
onset power ceiling that no rotation angle can exceed at this overlap.
Under $t_5$ innovations the uncorrected Gaussian nulls over-reject
($35$--$40\%$).
The $\sqrt{1+\widehat\kappa}$ scaling of Remark~\ref{rem:elliptical-null}
restores size ($6$--$10\%$) at a substantial power cost, whereas the
tail-adaptive bootstrap, which resamples from an elliptical $t$ with fitted
index $\widehat\nu=4+2/\widehat\kappa$, keeps size at $8\%$ while retaining
$50\%$ power at $\theta{=}1$; we recommend this configuration for heavy-tailed
data.
The covariance-level baseline $\|S_t-S_{t-1}\|_F$, calibrated by the same
parametric bootstrap (likelihood-ratio covariance-equality tests are infeasible
here, with $N(N{+}1)/2=1{,}275$ parameters against $2M=252$ observations), is
more powerful under Gaussian rotations because it reacts to any covariance
change, but its calibration fails badly under $t_5$ ($73\%$ size) while the
subspace tests remain usable.
In the scale-change design, finally, the subspace tests are correctly silent
($2$--$3\%$), answering whether the dominant eigenspace rotated rather than
whether the covariance changed.

\begin{table}[!ht]
  \centering
  \caption{Size and power at nominal $5\%$, flag rates by procedure
  ($150$--$200$ reps per cell; $\theta{=}0$ rows are size; ``Ell.''\ =
  $\sqrt{1+\widehat\kappa}$-scaled thresholds; ``Boot.-$t$'' = tail-adaptive
  bootstrap; $\|\Delta S\|_F$ = covariance-level baseline).}
  \label{tab:tier-power}
  \footnotesize
  \resizebox{\textwidth}{!}{\begin{tabular}{llllrrrrrrr}
\toprule
Gap & Dist. & Design & $\theta$ & $D^{\mathrm{true}}$ & Band $\alpha{=}0.75$ & Analytic & Bootstrap & Ell.\ (anl./bt.) & Boot.-$t$ & $\|\Delta S\|_F$ \\
\midrule
Strong & gauss & onset & 0 & 0.00 & 0.0\% & 0.5\% & 1.0\% & -- & -- & -- \\
Strong & gauss & onset & 0.05 & 0.07 & 0.0\% & 1.0\% & 1.0\% & -- & -- & -- \\
Strong & gauss & onset & 0.1 & 0.14 & 0.0\% & 3.0\% & 0.5\% & -- & -- & -- \\
Strong & gauss & onset & 0.2 & 0.28 & 0.0\% & 3.5\% & 0.5\% & -- & -- & -- \\
Strong & gauss & onset & 0.4 & 0.55 & 0.0\% & 11.0\% & 3.0\% & -- & -- & -- \\
Strong & gauss & post & 0 & 0.00 & 0.0\% & 2.0\% & 1.0\% & 2/1\% & -- & 6.0\% \\
Strong & gauss & post & 0.2 & 0.28 & 0.0\% & 1.3\% & 5.3\% & -- & -- & -- \\
Strong & gauss & post & 0.6 & 0.80 & 0.0\% & 56.7\% & 66.0\% & 57/66\% & -- & 84.0\% \\
Strong & gauss & post & 1 & 1.19 & 0.0\% & 93.3\% & 94.7\% & 93/95\% & -- & 99.3\% \\
Strong & gauss & post & 0 (scale 1.5) & 0.00 & 0.0\% & 2.7\% & 2.0\% & 3/2\% & 2.0\% & 6.7\% \\
Strong & t5 & onset & 0 & 0.00 & 0.0\% & 40.0\% & 34.7\% & 10/6\% & 8.0\% & 73.3\% \\
Strong & t5 & post & 1 & 1.19 & 62.7\% & 78.0\% & 96.7\% & 15/19\% & 50.0\% & 100.0\% \\
Weak & gauss & onset & 0 & 0.00 & 0.0\% & 1.0\% & 0.0\% & -- & -- & -- \\
Weak & gauss & onset & 0.4 & 0.55 & 0.0\% & 5.0\% & 1.0\% & -- & -- & -- \\
Weak & gauss & post & 1 & 1.19 & 7.3\% & 71.3\% & 80.0\% & -- & -- & -- \\
\bottomrule
\end{tabular}
}
\end{table}

\paragraph{Verification of the power law and frontier.}
Figure~\ref{fig:power-frontier} overlays the predicted power curves (linear and
mixture-exact signals, population spectrum) on the empirical rejection rates of
the analytic test in the two-window experiment.
For the experiment's configuration ($M{=}126$, $s{=}31$, $\phi\approx0.246$,
strong gap), the post-design frontier \eqref{eq:theta50} gives
$\theta_{50}=0.41$, bracketed by the empirical transition from $3\%$ rejection
at $\theta{=}0.2$ to $90\%$ at $\theta{=}0.6$ under oracle calibration.
The onset design violates the frontier condition (right side $1.59>1$), so the
first-order ceiling applies. The predicted maximum of about $21\%$ at
$\theta=\pi/4$ compares with $25\%$ observed, and the mixture-exact law tracks
the empirical rise beyond $\pi/4$ (ceiling $\approx35\%$ at the peak
$\theta\approx0.96$, against $40\%$ observed at $\theta{=}1$), where the
linearized signal $\phi\cos\theta\,D_{\mathrm{true}}$ degrades by design.
Two deviations are documented. The plug-in null is conservative (size
$0.5$--$2\%$ at nominal $5\%$, from Marchenko--Pastur spreading of window-A
eigenvalues), and at $s/M{=}0.25$ the oracle-calibrated size inflation of about
$11\%$ from second-order shared-block terms
shifts the onset column upward.
Corollary~\ref{cor:theta-min} reduces the frontier to two numbers.
The explicit formula \eqref{eq:theta-min-explicit} gives $\theta_{\min}=0.36$
and $n^{*}\approx55$ for this spectrum, against $\theta_{\min}=0.41$ and
$n^{*}\approx68$--$77$ from the exact-increment quantiles
\eqref{eq:theta-min-unified}; the closed-form constant is roughly $10\%$
optimistic because the exact-increment null is heavier-tailed than the weighted
$\chi^2$, and both versions classify the designs correctly, with
$s{=}31<n^{*}$ (onset ceiling) and $M{=}126>n^{*}$ (detectable).
These formulas make the detection delay visible in Table~\ref{tab:tier-power}
explicit. After a rotation of size $\theta$, detection at $50\%$ power requires
$n_{\mathrm{post}}\ge 8(\tilde q_{1-\alpha}-\tilde q_{50})/\sin^2(2\theta)$
post-break observations in the window, and fewer than $n^{*}$ post-break
observations rule out $50\%$ power at any angle.

\begin{figure}[!ht]
  \centering
  \includegraphics[width=\linewidth]{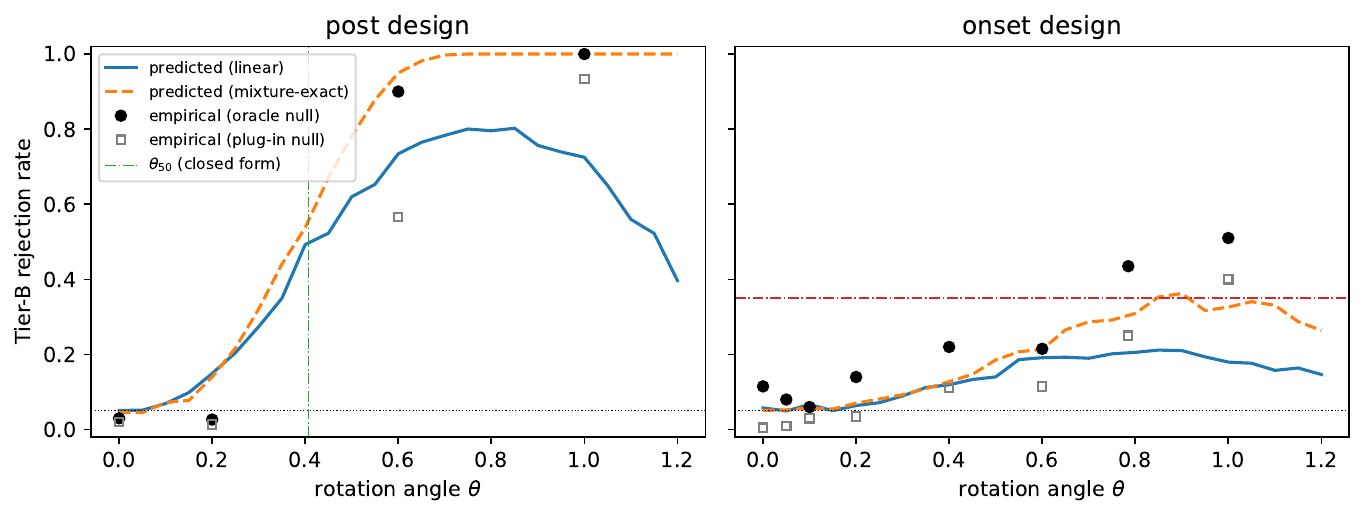}
  \caption{Power of the analytic test, predicted (linear and mixture-exact
  first-order laws, population spectrum) vs.\ empirical rejection rates
  (two-window experiment;
  oracle = population-eigenvalue null quantile, plug-in = window-A sample
  eigenvalues). The left panel shows the post design with the closed-form $\theta_{50}$ of
  \eqref{eq:theta50}, the right the onset design ($\phi\approx0.25$), where the
  frontier condition fails and the power ceiling applies.}
  \label{fig:power-frontier}
\end{figure}

\paragraph{The first-order null vs.\ the weighted-$\chi^2$ law.}
Figure~\ref{fig:null-vs-chi2} compares the simulated exact first-order null of
$\widehat D_{K,t}$ (Proposition~\ref{prop:noise-floor}) against the weighted-$\chi^2$
limit at a realistic overlap $s/M=0.25$ ($N{=}50$, $K{=}2$, $M{=}126$). The two agree
in the body, but the exact null is slightly heavier-tailed, with a $95\%$ quantile of
$0.431$ against the weighted-$\chi^2$ value $0.415$, a $4\%$ gap that, compounded by
the second-order shared-block terms of Proposition~\ref{prop:second-order}, produces
the $\approx10\%$ analytic-test size at this overlap (Table~\ref{tab:tier-power}) and
motivates the estimator-aware bootstrap. At the ends of the overlap range
($s/M\to0$ or $s/M=1$) the two laws coincide.

\begin{figure}[!ht]
  \centering
  \includegraphics[width=\linewidth]{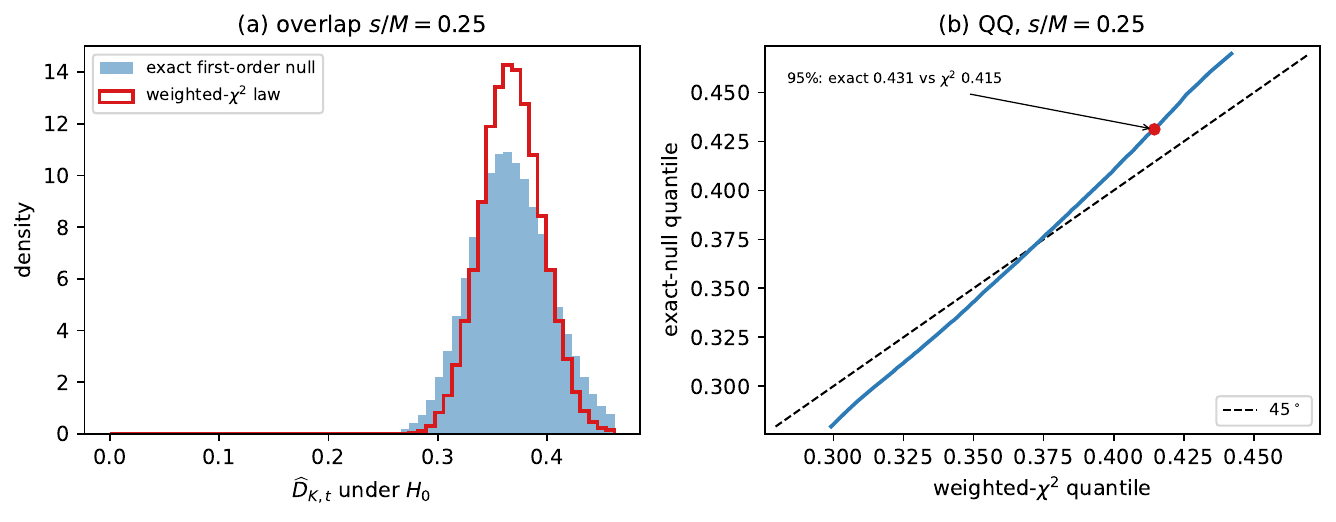}
  \caption{First-order null of the projector movement $\widehat D_{K,t}$ at overlap
  $s/M=0.25$ ($N{=}50$, $K{=}2$, $M{=}126$, $4\times10^4$ draws). (a) density of the
  exact first-order null (Proposition~\ref{prop:noise-floor}) vs.\ the weighted-$\chi^2$
  law $\sum_{i\le K<j} w_{ij}Z_{ij}^2$; (b) QQ plot of the two laws with the $95\%$
  quantiles marked. The exact null is mildly heavier-tailed at this intermediate
  overlap, the regime where the analytic test's size drifts upward.}
  \label{fig:null-vs-chi2}
\end{figure}

\paragraph{Multiplicity across dates.}
Per-date flags along a monitored path are serially dependent, because consecutive
windows share $M-s$ observations.
A dedicated experiment with persistent-regime paths ($40$ paths, plug-in
analytic-test $p$-values) finds that per-date size
nevertheless holds on pure-null dates ($2.6\%$ at nominal $5\%$), and that
Benjamini--Hochberg at $q{=}0.10$ across dates controls the realized false-discovery
proportion ($2.5\%$, conservative; Benjamini--Yekutieli $2.1\%$) at the cost of roughly
two-thirds of the naive procedure's true-discovery rate.
Formal FDR theory under this dependence remains open; the experiment indicates the
practical risk is conservatism, not anti-conservatism.

\subsection{Interval coverage for spectral-functional inference}
\label{sec:mc-si-coverage}

Table~\ref{tab:si-coverage} verifies Proposition~\ref{prop:scale-immunity} at
$M{=}252$ (500 reps; the spiked DGP of Table~\ref{tab:ar-coverage}).
The predicted inflation factor is confirmed. At $N{=}30$ under elliptical $t_8$
the empirical variances of $\widehat{\AR}_2$, $\widehat f_1$, and $\widehat H$
are $1.55$--$1.61$ times their Gaussian counterparts, against the theoretical
$1+\kappa=1.5$.
The contrast between level and ratio functionals is large.
The Gaussian-formula interval for total variance undercovers severely under
heavy tails ($49\%$ at $t_5$, $N{=}30$; $27\%$ at $N{=}115$) because it omits
the common-shock term $\kappa\langle G,\Sigma\rangle^2$, and the corrected
interval ($\widehat\kappa$-MLE with the full decomposition
\eqref{eq:si-variance}) restores $90$--$98\%$; for the scale-invariant
functionals the entire correction is the scalar factor $(1+\widehat\kappa)$,
and $\AR_2$ and $f_1$ reach $87$--$96\%$ under $t_5$ even at $N{=}115$.
Spectral entropy is the negative case, and we report it as such. With a flat bulk the gradient
$\nabla H\propto(\log p_i+H)$ degenerates exactly where Marchenko--Pastur
spreading biases the plug-in most, so first-order intervals undercover ($77\%$
Gaussian at $N{=}30$, $0\%$ at $N{=}115$).
The failure is the entropy analogue of the cleaning--debiasing wedge of
\S\ref{sec:ar-hd-caution}, and second-order or debiased entropy inference is left
open.

\begin{table}[!ht]
  \centering
  \caption{Coverage of nominal $95\%$ intervals for spectral functionals
  ($M{=}252$, 500 reps; spiked DGP of Table~\ref{tab:ar-coverage}; $t_8$ has
  $\kappa{=}0.5$ and $t_5$ has $\kappa{=}2$). ``Gaussian'' uses $\kappa{=}0$;
  ``$\widehat\kappa$-MLE'' the radial estimator of Remark~\ref{rem:kappa-mle}.
  The level functional $\mathrm{tr}$ additionally requires the common-shock term
  of \eqref{eq:si-variance}, which the scale-invariant rows shed by
  Proposition~\ref{prop:scale-immunity}.}
  \label{tab:si-coverage}
  \footnotesize
  \begin{tabular}{llrrrrrr}
\toprule
 & & \multicolumn{3}{c}{$N{=}30$} & \multicolumn{3}{c}{$N{=}115$} \\
\cmidrule(lr){3-5}\cmidrule(lr){6-8}
Functional & CI & gauss & $t_8$ & $t_5$ & gauss & $t_8$ & $t_5$ \\
\midrule
$\AR_2$ & Gaussian & 93\% & 87\% & 79\% & 88\% & 74\% & 45\% \\
 & $\widehat\kappa$-MLE & 93\% & 91\% & 93\% & 88\% & 86\% & 87\% \\
$f_1$ & Gaussian & 95\% & 90\% & 84\% & 94\% & 86\% & 67\% \\
 & $\widehat\kappa$-MLE & 95\% & 93\% & 95\% & 94\% & 95\% & 96\% \\
$H$ & Gaussian & 77\% & 57\% & 35\% & 0\% & 0\% & 0\% \\
 & $\widehat\kappa$-MLE & 77\% & 64\% & 66\% & 0\% & 0\% & 0\% \\
\addlinespace
$\mathrm{tr}$ (level) & Gaussian & 94\% & 67\% & 49\% & 96\% & 42\% & 27\% \\
 & $\widehat\kappa$-MLE & 94\% & 91\% & 90\% & 96\% & 97\% & 98\% \\
\bottomrule
\end{tabular}

\end{table}

\begin{table}[!ht]
  \centering
  \caption{Coverage of nominal $95\%$ AR$_2$ intervals ($M{=}252$, 500 reps; t$_5$ is
  elliptical with $\kappa{=}2$).}
  \label{tab:ar-coverage}
  \footnotesize
  \begin{tabular}{llrrrr}
\toprule
$N$ & Dist. & Method & Coverage & Mean width & Mean bias \\
\midrule
30 & gauss & Delta (Gaussian) & 93.6\% & 0.064 & +0.0061 \\
30 & gauss & Delta (moment $\hat\kappa$) & 93.8\% & 0.064 & +0.0061 \\
30 & gauss & Delta (radial-MLE $\hat\kappa$) & 93.6\% & 0.064 & +0.0061 \\
30 & gauss & Delta ($\kappa$ true) & 93.6\% & 0.064 & +0.0061 \\
30 & gauss & HD debiased (MLE $\hat\kappa$) & 94.8\% & 0.064 & +0.0007 \\
30 & t5 & Delta (Gaussian) & 76.2\% & 0.064 & +0.0130 \\
30 & t5 & Delta (moment $\hat\kappa$) & 89.6\% & 0.092 & +0.0130 \\
30 & t5 & Delta (radial-MLE $\hat\kappa$) & 90.2\% & 0.100 & +0.0130 \\
30 & t5 & Delta ($\kappa$ true) & 94.4\% & 0.111 & +0.0130 \\
30 & t5 & HD debiased (MLE $\hat\kappa$) & 91.8\% & 0.101 & +0.0086 \\
115 & gauss & Delta (Gaussian) & 90.4\% & 0.033 & +0.0076 \\
115 & gauss & Delta (moment $\hat\kappa$) & 90.4\% & 0.033 & +0.0076 \\
115 & gauss & Delta (radial-MLE $\hat\kappa$) & 90.4\% & 0.033 & +0.0076 \\
115 & gauss & Delta ($\kappa$ true) & 90.4\% & 0.033 & +0.0076 \\
115 & gauss & HD debiased (MLE $\hat\kappa$) & 96.2\% & 0.033 & +0.0003 \\
115 & t5 & Delta (Gaussian) & 46.4\% & 0.036 & +0.0248 \\
115 & t5 & Delta (moment $\hat\kappa$) & 61.4\% & 0.052 & +0.0248 \\
115 & t5 & Delta (radial-MLE $\hat\kappa$) & 87.6\% & 0.079 & +0.0248 \\
115 & t5 & Delta ($\kappa$ true) & 75.0\% & 0.062 & +0.0248 \\
115 & t5 & HD debiased (MLE $\hat\kappa$) & 91.2\% & 0.080 & +0.0202 \\
\bottomrule
\end{tabular}

\end{table}

\paragraph{Necessity of the elliptical assumption.}
The $(1+\kappa)$ correction calibrates the scalar-functional intervals only within the
elliptical family. Table~\ref{tab:non-elliptical} stress-tests it on skewed,
\emph{non}-elliptical innovations (per-coordinate log-normal, and an asymmetric
two-component mixture), alongside Gaussian and elliptical-$t_5$ references, for the
spiked design of Table~\ref{tab:ar-coverage}. Under the elliptical references the
radial-MLE $\widehat\kappa$ interval recovers coverage (e.g.\ $t_5$ at $N{=}115$,
$45\%\!\to\!87\%$); under the skewed DGPs it does not (log-normal $37\%\!\to\!48\%$;
mixture $\approx\!65\%$), because skewness injects third-moment terms that no kurtosis
factor can absorb. The elliptical assumption therefore does real work, and
genuinely non-elliptical data call for the estimator-aware bootstrap or a
distribution-free fallback.

\begin{table}[!ht]
  \centering
  \caption{Coverage of nominal $95\%$ $\AR_K$ intervals under elliptical vs.\
  non-elliptical innovations ($M{=}252$, $K{=}2$, $500$ reps; seed 2024; spiked
  spectrum of Table~\ref{tab:ar-coverage}). ``Gaussian CI'' uses $\kappa{=}0$;
  ``$\widehat\kappa$-MLE CI'' uses the radial-MLE estimate. The correction restores
  coverage for the elliptical rows ($t_5$) but not for the skewed, non-elliptical rows
  (log-normal, mixture).}
  \label{tab:non-elliptical}
  \footnotesize
  \begin{tabular}{lcccc}
\toprule
& \multicolumn{2}{c}{$N=30$} & \multicolumn{2}{c}{$N=115$} \\
\cmidrule(lr){2-3}\cmidrule(lr){4-5}
DGP & Gaussian CI & $\widehat\kappa$-MLE CI & Gaussian CI & $\widehat\kappa$-MLE CI \\
\midrule
Gaussian (elliptical, $\kappa{=}0$) & 93\% & 93\% & 87\% & 87\% \\
$t_5$ (elliptical, $\kappa{=}2$) & 78\% & 91\% & 45\% & 87\% \\
Log-normal (skewed) & 36\% & 62\% & 37\% & 48\% \\
Asym.\ mixture (skewed) & 66\% & 72\% & 64\% & 65\% \\
\bottomrule
\end{tabular}

\end{table}

\paragraph{Verification of the wedge formulas.}
As a diagnostic for the absorption ratio, Figure~\ref{fig:ar-wedge}
overlays the predicted wedge curves $+\Delta_S$ and $-\Delta_O$ as functions
of $c$ on the empirical mean relative $\AR_2$ errors at
$N\in\{32,63,115,202\}$, $M{=}252$ ($500$ reps).
The predictions are also sharp pointwise; at $c\approx0.456$,
$\Delta_S\cdot\AR_2$ predicts a bias of $+0.0077$ against the measured
$+0.0076$ of Table~\ref{tab:ar-coverage}.
Sample and oracle track the predictions across the entire range
$c\in[0.13,0.80]$; at $c{=}0.456$ the predictions are $+5.4\%$ and $-5.1\%$
against empirical values of $+5.2\%$ and $-5.2\%$.
The QIS plug-in lands below the oracle wedge ($-13.3\%$ against $-5.2\%$ at
$c{=}0.456$) because finite-sample QIS shrinks isolated spikes beyond the oracle limit
$\lambda\alpha^2+\sigma^2(1-\alpha^2)$, so the oracle wedge of
Proposition~\ref{prop:wedge} is a lower bound on the practical distortion.
The trace-preserving debiased plug-in of part~(iii) is unbiased to within
simulation error through $c\approx0.46$ (residual at most $0.12\%$), with $1.7\%$
remaining at $c{=}0.80$; a numerator--denominator decomposition attributes this
residual entirely to edge mis-detection at extreme concentration, where
upper-bulk eigenvalues cross the estimated Marchenko--Pastur edge and bias
$\widehat\sigma^2$ down and $\psi^{-1}$ up, while the trace-preserving
denominator stays unbiased (relative error below $3\times10^{-4}$) at every
$c$, as part~(iii) predicts.
The naive QIS-plug-in Wald interval deteriorates exactly as the wedge predicts,
with coverage $86\%$, $65\%$, $26\%$, and $11\%$ along the same grid, as the
deterministic wedge swamps the $O(M^{-1/2})$ interval width.

\begin{figure}[!ht]
  \centering
  \includegraphics[width=0.78\linewidth]{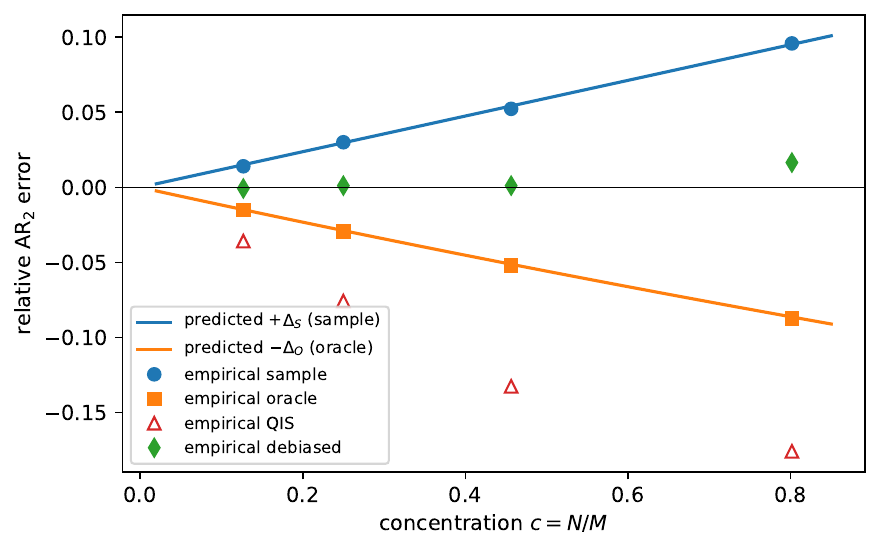}
  \caption{Cleaning--debiasing wedge across concentration $c=N/M$
  (spikes $(6,3.5,1.2)$, bulk $\sigma^2{=}0.5$, $K{=}2$, $M{=}252$, $500$ reps).
  Curves give the predicted $+\Delta_S$ (sample) and $-\Delta_O$ (Frobenius oracle)
  of Proposition~\ref{prop:wedge}; markers are the empirical mean relative
  $\AR_2$ errors. QIS overshoots the oracle wedge; the trace-preserving
  spike-debiased plug-in is unbiased to MC precision through $c\approx0.46$.}
  \label{fig:ar-wedge}
\end{figure}

\subsection{Computational cost}
\label{sec:mc-runtime}

Table~\ref{tab:mc-runtime} reports mean wall-clock time per window for the
covariance fit and one monitor step.

\begin{table}[!ht]
  \centering
  \caption{Mean wall-clock time per window (ms) for the covariance fit and one monitor
  step (seed 99, 30 replications per cell).}
  \label{tab:mc-runtime}
  \footnotesize
  \begin{tabular}{rrlrr}
\toprule
$N/M$ & $M$ & Estimator / task & ms (cov) & ms (monitor) \\
\midrule
1.49 & 77 & Sample & 0.2 & --- \\
1.49 & 77 & LW & 3.0 & --- \\
1.49 & 77 & QIS & 1.9 & 4.4 \\
1.49 & 77 & MP-clip & 1.9 & --- \\
1.00 & 115 & Sample & 0.2 & --- \\
1.00 & 115 & LW & 3.8 & --- \\
1.00 & 115 & QIS & 2.3 & 4.9 \\
1.00 & 115 & MP-clip & 2.1 & --- \\
0.50 & 230 & Sample & 0.2 & --- \\
0.50 & 230 & LW & 4.4 & --- \\
0.50 & 230 & QIS & 2.3 & 5.0 \\
0.50 & 230 & MP-clip & 2.0 & --- \\
0.25 & 460 & Sample & 0.3 & --- \\
0.25 & 460 & LW & 4.6 & --- \\
0.25 & 460 & QIS & 2.5 & 6.1 \\
0.25 & 460 & MP-clip & 2.3 & --- \\
\bottomrule
\end{tabular}

\end{table}

Table~\ref{tab:tier-runtime} adds the cost of the calibrated tests per
evaluation date. The worst-case band is essentially free ($0.02$\,ms), the
analytic null costs $32$--$87$\,ms and requires no eigendecompositions, and the
QIS-aware bootstrap costs $0.1$--$0.8$\,s, fast enough for daily monitoring of
hundreds of assets; by the equivariance argument of \S\ref{sec:rot-equiv} the
sample-estimator bootstrap halves that cost.

\begin{table}[!ht]
  \centering
  \caption{Wall-clock cost per evaluation date (single core).}
  \label{tab:tier-runtime}
  \footnotesize
  \begin{tabular}{llr}
\toprule
$(N,M)$ & Procedure & ms per date \\
\midrule
(50,126) & Band (capped Davis--Kahan) & 0.0 \\
(50,126) & Analytic (400 draws) & 31.6 \\
(50,126) & Bootstrap (QIS, $B{=}60$) & 121.4 \\
(50,126) & Bootstrap (QIS, $B{=}99$) & 200.6 \\
(50,126) & Bootstrap (sample, $B{=}60$) & 55.3 \\
(50,126) & Bootstrap (sample, $B{=}99$) & 92.7 \\
(115,252) & Band (capped Davis--Kahan) & 0.0 \\
(115,252) & Analytic (400 draws) & 86.7 \\
(115,252) & Bootstrap (QIS, $B{=}60$) & 509.4 \\
(115,252) & Bootstrap (QIS, $B{=}99$) & 815.2 \\
(115,252) & Bootstrap (sample, $B{=}60$) & 256.4 \\
(115,252) & Bootstrap (sample, $B{=}99$) & 399.2 \\
\bottomrule
\end{tabular}

\end{table}

\newpage
\section{Discussion}
\label{sec:discussion}
\label{sec:conclusion}

Movements in the spectral functionals and the dominant eigenspace of a rolling
covariance estimate are routinely treated as structural change, yet both
fluctuate under estimation noise, and shrinkage alters the noise law.
This paper supplies the calibration that distinguishes the two.
For the eigenspace, a distribution-free Davis--Kahan band gauges identification,
a first-order null law calibrates observed movement between overlapping windows
and transfers without change to rotation-equivariant shrinkage estimators, and a
local power law fixes the smallest detectable rotation and the number of fresh
observations a two-window monitor needs to see it; a parametric bootstrap covers
the cases the analytic null does not.
The same mechanism calibrates the scalar functionals. First-order immunity to
elliptical kurtosis holds for scale-invariant functionals and only for them, so a
single estimated scalar calibrates the absorption ratio,
the leading share, and the projector null together, and window-level volatility
cancels exactly.
The same machinery exposes a high-dimensional pitfall, that shrinkage cleaning
mis-centers the absorption ratio by an explicit margin, and supplies the
spike-debiased correction that removes it.
The evidence throughout is simulation under a known population covariance; the
equity-panel appendix shows what the procedures produce when the population is
unknown, where their output is a calibrated diagnostic rather than a test of
latent change.

\paragraph{Limitations and open problems.}
The sequence of flags is serially dependent because consecutive windows share
$M-s$ observations; empirically, per-date size holds on pure-null dates
and Benjamini--Hochberg across dates is conservative (\S\ref{sec:tier-power}),
but formal false-discovery theory under this dependence is open.
The analytic test has a documented size distortion at intermediate overlap
($s/M\approx 0.25$), bounded by Proposition~\ref{prop:second-order} when
$\eta<\Delta/4$.
Fully calibrated absorption-ratio inference under elliptical heavy tails in
high dimensions remains open (\S\ref{sec:ar-hd-caution}), and the radial-MLE
kurtosis estimator of Remark~\ref{rem:kappa-mle} is conservative for extreme
tails by construction.
First-order entropy intervals degenerate near flat spectra, where $\nabla H\to
0$ while the plug-in bias peaks (Table~\ref{tab:si-coverage}); second-order or
debiased entropy inference is likewise open.
The $(1+\kappa)$ immunity factor is a fixed-$N$ statement. At $N/M\approx 0.46$
under $t_5$ the empirical inflation exceeds it, and the conservative
$\widehat\kappa$-MLE absorbs the excess only partially.
The monitoring dimension $K$ and the crisis thresholds are fixed by
configuration.

\paragraph{Hyperparameters and reproducibility.}
The monitoring dimension is selected by the dominant eigengap
(Proposition~\ref{prop:stable-K}, capped at $K_{\max}$); the operational scale
$\alpha^\star$ is the smallest grid value whose simulated calm-period flag rate stays
at or below the nominal level (\S\ref{sec:mc-alpha}); and the window length $M$ and
overlap $s$ are reported with each experiment. All Monte Carlo results use fixed seeds
and the replication counts stated in the captions; the parametric bootstrap
(Algorithm~\ref{alg:bootstrap}) resamples $B$ Gaussian (or fitted-elliptical) paths
per evaluation date.

\paragraph{Data availability.}
The equity panel comprises daily adjusted returns for $115$ US large-cap stocks
retrieved from Yahoo Finance through the \texttt{yfinance} interface.

\appendix

\newpage
\section{Proofs}
\label{app:proofs-detailed}

This appendix collects the proofs cited from the main text.

\subsection[Projectors and the $\sin\Theta$ identity]{Projectors and the $\sin\Theta$ identity (Lemma~\ref{lem:proj-sin})}

Let $U,\widehat U\in\R^{N\times K}$ have orthonormal columns ($U^\top U=I_K$,
$\widehat U^\top\widehat U=I_K$), and define $P=UU^\top$,
$\widehat P=\widehat U\widehat U^\top$, and $M=\widehat U^\top U$.
Expanding the Frobenius norm gives
\[
  \|P-\widehat P\|_F^2
  =\mathrm{tr}\bigl[(UU^\top-\widehat U\widehat U^\top)^\top
  (UU^\top-\widehat U\widehat U^\top)\bigr].
\]
By the cyclic property $\mathrm{tr}(AB)=\mathrm{tr}(BA)$ and the idempotence relations
$P^2=P$ and $\widehat P^2=\widehat P$, this becomes
\[
  \|P-\widehat P\|_F^2
  =\mathrm{tr}(P)+\mathrm{tr}(\widehat P)-2\mathrm{tr}(\widehat P P).
\]
Rank-$K$ projectors have trace $K$, so $\mathrm{tr}(P)=\mathrm{tr}(\widehat P)=K$, while
$\mathrm{tr}(\widehat P P)=\mathrm{tr}(\widehat U\widehat U^\top U U^\top)
=\mathrm{tr}(U^\top\widehat U\widehat U^\top U)=\mathrm{tr}(M M^\top)=\|M\|_F^2$.
Hence $\|P-\widehat P\|_F^2=2K-2\|M\|_F^2$.
If $M$ has singular values $\sigma_1,\ldots,\sigma_K$ (the principal cosines), then
$\|M\|_F^2=\sum_{j=1}^K\sigma_j^2=\sum_{j=1}^K\cos^2\theta_j$, while by definition
$\|\sin\Theta\|_F^2=\sum_{j=1}^K\sin^2\theta_j
=K-\sum_{j=1}^K\cos^2\theta_j=K-\|M\|_F^2$.
Therefore $\|P-\widehat P\|_F^2=2\|\sin\Theta\|_F^2$, as claimed.

\subsection[Weyl's inequality]{Weyl's inequality (Lemma~\ref{lem:weyl})}

For symmetric $C$, the Courant--Fischer characterization gives
\[
  \lambda_i(C)=\min_{\dim(L)=N-i+1}\;\max_{x\in L,\,\|x\|=1}x^\top Cx .
\]
For each unit vector $x$ we have $x^\top Cx\le x^\top\Sigma x+\eta$; maximizing over
$x\in L$ and then minimizing over subspaces $L$ with $\dim(L)=N-i+1$ yields
$\lambda_i(C)\le\lambda_i(\Sigma)+\eta$. Exchanging the roles of $C$ and $\Sigma$ gives
the reverse inequality.

\subsection[The overlapping-window noise floor]{The overlapping-window noise floor (Proposition~\ref{prop:noise-floor})}\label{app:noise-floor}

The argument has three ingredients. The shared observations cancel exactly in the
covariance increment, the projector difference is linearized by the Kato expansion, and
the variance of the linear term is computed entrywise in the population eigenbasis.
Write $S_t=\Sigma+E_t$, $S_{t-1}=\Sigma+E_{t-1}$ for the two window sample covariances
and $\Delta E=E_t-E_{t-1}=S_t-S_{t-1}$.

With windows $\{1,\ldots,M\}$ and $\{s+1,\ldots,s+M\}$,
$\Delta E=\frac1M\bigl(\sum_{k=M+1}^{M+s}r_kr_k^\top-\sum_{k=1}^{s}r_kr_k^\top\bigr)$, and
the $M-s$ shared outer products cancel exactly.
The Kato expansion of the top-$K$ projector gives
$\widehat P_{K,t}-P_K=L(E_t)+Q(E_t)$ with linear part
$L(E)=\sum_{i\le K<j}\frac{u_i^\top E u_j}{\lambda_i-\lambda_j}
(u_iu_j^\top+u_ju_i^\top)$ and quadratic remainder $Q$.
Subtracting the two dates yields $\widehat P_{K,t}-\widehat P_{K,t-1}
=L(\Delta E)+\{Q(E_t)-Q(E_{t-1})\}$.

To compute the variance of the increment entries, pass to the population eigenbasis,
where $y_k=U^\top r_k$ has independent components $y_{ik}\sim N(0,\lambda_i)$. Then
$g_{ij}:=u_i^\top\Delta E\,u_j=\frac1M\sum_{k\,\mathrm{new}}y_{ik}y_{jk}
-\frac1M\sum_{k\,\mathrm{old}}y_{ik}y_{jk}$
has mean zero and variance $2s\lambda_i\lambda_j/M^2$ for $i\ne j$, since each of the
$2s$ independent products has variance $\lambda_i\lambda_j$.
Because each pair contributes two symmetric entries,
$\|L(\Delta E)\|_F^2=2\sum_{i\le K<j}g_{ij}^2/(\lambda_i-\lambda_j)^2$, and therefore
$\E\|L(\Delta E)\|_F^2=\frac{4s}{M^2}\sum_{i\le K<j}
\frac{\lambda_i\lambda_j}{(\lambda_i-\lambda_j)^2}$, which is \eqref{eq:noise-floor};
the exact first-order law \eqref{eq:exact-null} restates $g_{ij}$ as a sum of $2s$
Gaussian products.

Two limiting regimes confirm the formula. At $s=M$ the two windows are disjoint and the
mean is twice the one-sample value
\[
  \frac{2}{M}\sum_{i\le K<j}\frac{\lambda_i\lambda_j}{(\lambda_i-\lambda_j)^2}
\]
of \citet{KoltchinskiiLounici2017}.
As $s\to\infty$, $g_{ij}/\sqrt{2s\lambda_i\lambda_j/M^2}\to_d N(0,1)$ with vanishing
cross-pair dependence, giving the weighted-chi-square limit.
The neglected $Q$-difference contributes second-order terms shared between the two
windows and is absorbed by the estimator-aware bootstrap. A simulation check shows that
its effect is visible at intermediate overlap, with empirical size $\approx 10\%$
at $s/M=0.25$.

\emph{Centering and normalization.}
The statement and the exact-increment simulator are written for known-mean
sampling with the normalization $1/M$, whereas the implemented monitor demeans
each window and divides by $M-1$ (\S\ref{sec:setting}).
Demeaning perturbs each window covariance by a rank-one term of order $1/M$, and
the rescaling $M/(M-1)$ is common to both windows, so both effects are second
order in the increment and cancel from the first-order law.
The simulation studies of \S\ref{sec:mc} evaluate the demeaned $(M-1)$-statistic
against the known-mean null, so the reported sizes already include any residual
mismatch.

\subsection[Contour bound on the remainder]{Contour bound on the remainder (Proposition~\ref{prop:second-order})}\label{app:second-order}

The proof bounds the quadratic remainder of the Riesz projector by a Neumann-series
estimate on a contour that separates the top-$K$ eigenvalues from the rest of the
spectrum. Let $\Gamma$ be the circle of center $(\lambda_1+\lambda_K)/2$ and radius
$r=(\lambda_1-\lambda_K+\Delta)/2$, so that $\Gamma$ encloses
$\{\lambda_1,\ldots,\lambda_K\}$ and
$\mathrm{dist}(\Gamma,\mathrm{spec}(\Sigma))\ge\Delta/2$, whence
$\sup_{z\in\Gamma}\|R(z)\|_{\op}\le 2/\Delta$ for $R(z)=(z I-\Sigma)^{-1}$.

If $\eta=\|E\|_{\op}<\Delta/4$, every eigenvalue of $\Sigma+E$ stays at distance
$\ge\Delta/4$ from $\Gamma$ by Weyl's inequality, and the perturbed resolvent admits
the Neumann series $R_E(z)=\sum_{k\ge0}R(z)\,[E R(z)]^k$ on $\Gamma$, convergent since
$\|ER(z)\|_{\op}\le 2\eta/\Delta<1/2$.
The Riesz representation gives $\widehat P_K-P_K=\frac{1}{2\pi i}\oint_\Gamma
[R_E(z)-R(z)]\,dz$, and a residue computation in the eigenbasis identifies the $k=1$
term with the linear map $L(E)$ of Proposition~\ref{prop:noise-floor}.
The remainder $Q(E)=\widehat P_K-P_K-L(E)$ collects the terms with $k\ge2$, so
\[
  \|Q(E)\|_{\op}
  \le\frac{\mathrm{len}(\Gamma)}{2\pi}\sum_{k\ge2}
  \Bigl(\frac{2}{\Delta}\Bigr)^{k+1}\eta^k
  = \frac{2r}{\Delta}\cdot
  \frac{(2\eta/\Delta)^2}{1-2\eta/\Delta}.
\]
To pass from operator to Frobenius norm, note that
$\mathrm{rank}(\widehat P_K-P_K)\le 2K$ as a difference of two rank-$K$ projectors, and
$\mathrm{rank}(L(E))\le 2K$ because the sum over $i\le K$ of symmetrized rank-2 terms
supported on $\mathrm{span}\{u_i\}\oplus\mathrm{span}\{u_j\}_{j>K}$ has column space of
dimension $\le 2K$; hence $\mathrm{rank}(Q)\le 4K$ and
$\|Q\|_F\le 2\sqrt K\,\|Q\|_{\op}$.
Applying the reverse triangle inequality to
$\widehat P_{K,t}-\widehat P_{K,t-1}=L(\Delta E)+Q(E_t)-Q(E_{t-1})$ gives
\eqref{eq:second-order}; the final form uses $(1-2\eta/\Delta)^{-1}\le2$ for
$\eta<\Delta/4$.
The size statement follows by intersecting with the event
$\{\eta_t\vee\eta_{t-1}<h\}$ and a union bound.

\subsection{Consistency of the bootstrap}\label{app:bootstrap}

The strategy is a coupling argument. The bootstrap law is realized as a continuous
function of the plug-in covariance, so consistency of the plug-in transfers to the
bootstrap quantiles. Fix Gaussian arrays $Z\in\R^{(M+s)\times N}$ and represent
the two-window data under covariance $A$ as $X=ZA^{1/2}$; the statistic
$\widehat D_{K,t}=\varphi(Z;A)$ is, for fixed $Z$, a composition of (a) sample
covariances of the two windows (polynomial in $A^{1/2}$), and (b) top-$K$ spectral
projectors, which are continuous in the matrix argument wherever
$\lambda_K>\lambda_{K+1}$ for both window sample covariances.

For fixed $N\le M$ the sample covariance of a Gaussian window has almost surely simple
eigenvalues, so for almost every $Z$ the map $A\mapsto\varphi(Z;A)$ is continuous at
$A=\Sigma$. Hence $A_m\to\Sigma$ implies $\varphi(Z;A_m)\to\varphi(Z;\Sigma)$ a.s.,
i.e.\ $F_{A_m}\Rightarrow F_\Sigma$ by dominated convergence on bounded continuous test
functions.
Weak convergence plus continuity and strict increase of $F_\Sigma$ at $q_{95}(\Sigma)$
give $q_{95}(A_m)\to q_{95}(\Sigma)$; applying this along subsequences of
$C_M\to_p\Sigma$ yields $q_{95}(C_M)\to_p q_{95}(\Sigma)$.
For the size, $\mathbb{P}(\widehat D>q_{95}(C_M))
=\mathbb{P}(\widehat D>q_{95}(\Sigma))+o(1)\to0.05$ by Slutsky and continuity of
$F_\Sigma$ at $q_{95}(\Sigma)$; dependence between $\widehat D$ and $q_{95}(C_M)$ is
immaterial because the threshold converges to a constant.
Finally, the hypothesis $C_M\to_p\Sigma$ holds for QIS/LW at fixed $N$, since the
sample covariance is consistent and the shrinkage intensity vanishes as $N/M\to0$
\citep{LedoitWolf2004,LedoitWolf2022QIS}.

\subsection[Local power and the detectability frontier]{Local power and the detectability frontier (Proposition~\ref{prop:tier-b-power}, Corollary~\ref{cor:theta-min})}\label{app:tier-b-power}

The proof computes the mean of the mixture window exactly, derives the noncentral
first-order law of the projector increment under the local alternative, and then
obtains the frontier and the ceiling from the resulting median shift. Throughout we work in
the eigenbasis of $\Sigma_1=\mathrm{diag}(\lambda)$; the rotation acts on
coordinates $(1,b)$ with $b>K$.

\emph{(i) Mixture mean and attenuation.}
Write $c=\cos\theta$, $s=\sin\theta$.
Window B averages $M-s$ outer products with mean $\Sigma_1$ and $s$ with mean
$\Sigma_2$, so $\E S_B=(1-\phi)\Sigma_1+\phi\Sigma_2=\bar\Sigma$ exactly, and the
estimable signal is $D_{\mathrm{mix}}=\|P_K(\bar\Sigma)-P_K(\Sigma_1)\|_F$.
The perturbation $\Sigma_2-\Sigma_1$ has a single cross-cut entry, at $(1,b)$, equal
to $sc(\lambda_1-\lambda_b)$, so the linear map $L$ of
Proposition~\ref{prop:noise-floor} gives
$L(\Sigma_2-\Sigma_1)=sc\,(u_1u_b^\top+u_bu_1^\top)$ with
$\|L(\Sigma_2-\Sigma_1)\|_F=\sqrt2\,sc$, the gap $\lambda_1-\lambda_b$ cancelling.
Since $\bar\Sigma-\Sigma_1=\phi(\Sigma_2-\Sigma_1)$,
\[
  D_{\mathrm{mix}}=\phi\,\|L(\Sigma_2-\Sigma_1)\|_F+O(\theta^2)
  =\phi\cos\theta\cdot\sqrt2\sin\theta+O(\theta^2)
  =\phi\cos\theta\,D_{\mathrm{true}}+O(\theta^2).
\]
At $\theta=\pi/2$ the cross-cut entry vanishes and the mixture is diagonal,
with entries $(1-\phi)\lambda_1+\phi\lambda_b$ and
$(1-\phi)\lambda_b+\phi\lambda_1$ in coordinates $1$ and $b$.
If the mixed diagonal preserves the top-$K$ ordering, for which $\phi<\tfrac12$
together with
$\phi(\lambda_1-\lambda_b)<\min(\lambda_K-\lambda_b,\lambda_1-\lambda_{K+1})$
is sufficient, the top-$K$ projector is unchanged and $D_{\mathrm{mix}}=0$.
Once the ordering breaks, the top-$K$ set exchanges $u_1$ for $u_b$; at
$\phi=1$ the mixture equals $\Sigma_2$ and
$D_{\mathrm{mix}}(1,\theta)=\sqrt2\,\sin\theta$ for every $\theta$.

\emph{(ii) Noncentral first-order law.}
Under the gap condition the contour argument of Proposition~\ref{prop:second-order}
justifies linearizing both projectors at $\Sigma_1$, so that
$\widehat P_B-\widehat P_A=L(S_B-S_A)+O(\eta^2/\Delta^2)$ terms.
The increment $G=S_B-S_A$ retains the shared-block cancellation, with
$\E G=\phi(\Sigma_2-\Sigma_1)$, whose only off-block entry is
$\E g_{1b}=\phi sc(\lambda_1-\lambda_b)$.
Entry variances equal their null values up to relative $O(\theta^2)$
corrections from the $\Sigma_2$ blocks (e.g.\
$\mathrm{Var}_2(y_1y_b)=\bar\sigma_{11}\bar\sigma_{bb}+\bar\sigma_{1b}^2$),
which contribute to $\E\widehat D^2$ at order $s\theta^2/M^2\cdot
\lambda(\lambda_1-\lambda_b)/\Delta^2$, i.e.\ $O(\theta^2/M)$ relative to the
$O(1)$ noise floor and the $O(\phi^2\theta^2)$ mean term at fixed $\phi$.
Squaring the linear term gives the noncentral form with noncentrality
$\sum\tilde w\delta^2=\phi^2\|L(\Sigma_2-\Sigma_1)\|_F^2=\tfrac12\phi^2\sin^2 2\theta$
by (i), which is \eqref{eq:ncp}.

\emph{(iii) Frontier.}
The signal occupies the single pair coordinate $(1,b)$, so
$\widehat D^2_\theta=Q_0+2\psi\sqrt{\tilde w_{1b}}\,Z_{1b}+\psi^2$ with
$\psi^2$ the noncentrality and $Q_0$ the null statistic.
Since the tilt is mean-zero,
$\mathrm{med}(\widehat D^2_\theta)\approx q^2_{50}+\psi^2$; setting this equal
to $q^2_{1-\alpha}$ and solving
$\tfrac12\phi^2\sin^22\theta_{50}=q^2_{1-\alpha}-q^2_{50}$ yields
\eqref{eq:theta50}.
Under the weighted-$\chi^2$ approximation all null quantiles of
$\widehat D^2$ scale as $s/M^2$, so
$\sin 2\theta_{50}\propto\sqrt{s}/(M\phi)$, and with $\phi=s/M$ this is
$1/\sqrt{s}$ (onset) and $1/\sqrt{M}$ (post), i.e.\
$\theta_{50}\asymp C/\sqrt{n_{\mathrm{post}}}$ and the ratio $\sqrt{M/s}$.
Corollary~\ref{cor:theta-min} follows by substituting
$q_p^2=(4n_{\mathrm{post}}/M^2)\,\tilde q_p$ and $\phi=n_{\mathrm{post}}/M$
into \eqref{eq:theta50}, which gives
$\sin^2(2\theta_{\min})=8(\tilde q_{1-\alpha}-\tilde q_{50})/n_{\mathrm{post}}$,
which is \eqref{eq:theta-min-unified}; the explicit form
\eqref{eq:theta-min-explicit} applies the two-moment normal approximation
$\tilde q_p\approx\sum w_{ij}+z_p\,(2\sum w_{ij}^2)^{1/2}$ to
$Q=\sum w_{ij}Z_{ij}^2$, where the means cancel in the quantile \emph{gap} and the
median contributes $z_{50}=0$, leaving
$\tilde q_{1-\alpha}-\tilde q_{50}\approx z_{1-\alpha}(2\sum w_{ij}^2)^{1/2}$;
the constraint $\sin(2\theta_{\min})\le1$ rearranges to
$n_{\mathrm{post}}\ge n^{*}=8\sqrt2\,z_{1-\alpha}(\sum w_{ij}^2)^{1/2}$.

\emph{(iv) Ceiling.}
The bound $\sin^2 2\theta\le1$ caps \eqref{eq:ncp} at
$\phi^2/2$, attained at $\pi/4$; if $\phi^2/2<q^2_{1-\alpha}-q^2_{50}$ the
median of $\widehat D^2_\theta$ stays below the threshold for every $\theta$.
Non-perturbatively, $\widehat P_B$ concentrates on $P_K(\bar\Sigma)$ (fixed
$N$, $M\to\infty$), so the available signal is $D_{\mathrm{mix}}$ of (i);
replacing $\psi^2$ by $\max_\theta D^2_{\mathrm{mix}}$ in the median-shift
argument gives the persistent ceiling condition.

\subsection[Lipschitz gradient bounds]{Lipschitz gradient bounds (Proposition~\ref{prop:propagate})}\label{app:propagate}

Write $\lambda\in\R^N_+$ with $T(\lambda)=\sum_j\lambda_j$ and $p_i(\lambda)=\lambda_i/T(\lambda)$.
For the leading share $f_1(\lambda)=\lambda_1/T$, the chain rule gives, whenever
$\lambda_1>0$ and $T>0$,
\[
  \frac{\partial f_1}{\partial \lambda_i}
  =\frac{\delta_{i1}T-\lambda_1}{T^2},
  \qquad i=1,\ldots,N.
\]
On any compact set with $\lambda_i\ge\varepsilon>0$ and $T\ge\varepsilon N$ the
gradient is bounded, $\|\nabla f_1\|_2\le C_{f_1}<\infty$.
For the absorption ratio $\AR_K(\lambda)=\bigl(\sum_{i\le K}\lambda_i\bigr)/T$, the same
computation yields
\[
  \frac{\partial \AR_K}{\partial \lambda_i}
  =\frac{\mathbf{1}_{\{i\le K\}}T-\bigl(\sum_{j\le K}\lambda_j\bigr)}{T^2},
\]
which is bounded on the same compact set.
For the entropy $H(\lambda)=-\sum_i p_i\log p_i$, differentiating through
$p_i=\lambda_i/T$ gives
\[
  \frac{\partial H}{\partial \lambda_i}
  =-\sum_{j=1}^N\frac{\partial p_j}{\partial\lambda_i}(\log p_j+1),
\]
with $\partial p_j/\partial\lambda_i=(\delta_{ij}T-\lambda_j)/T^2$, and all terms are
bounded when $p_j\ge\varepsilon/T$.
Finally, for any $\lambda,\tilde\lambda$ in that set the mean-value theorem gives
$|g(\tilde\lambda)-g(\lambda)|\le \sup_\xi\|\nabla g(\xi)\|_2\|\tilde\lambda-\lambda\|_2$,
which gives the Lipschitz constant $L_{g,N}$ of Proposition~\ref{prop:propagate}.

\subsection[Kurtosis immunity and scale invariance]{Kurtosis immunity and scale invariance (Proposition~\ref{prop:scale-immunity})}\label{app:scale-immunity}

The proof records the delta-method variance as the variance of a scalar quadratic
form, derives the immunity direction from Euler's relation, and finally proves
the converse by integrating the gradient condition along rays.
Write $S_M$ for the sample covariance of $M$ i.i.d.\ elliptical observations with
covariance $\Sigma$ and kurtosis parameter $\kappa$.

\emph{(i) Variance formula.}
For $W$ continuously differentiable at $\Sigma$ with symmetric gradient $G$, the delta
method gives $\sqrt{M}(W(S_M)-W(\Sigma))=\sqrt{M}\,\langle G,S_M-\Sigma\rangle+o_p(1)$.
Since $\langle G,S_M\rangle=\mathrm{tr}(GS_M)=M^{-1}\sum_\tau r_\tau^\top G r_\tau$ is an
average of i.i.d.\ scalars, $\sqrt{M}(W(S_M)-W(\Sigma))\to_d N(0,\sigma^2_W)$ with
$\sigma^2_W=\mathrm{Var}(r^\top G r)$.
Expanding by bilinearity and inserting the elliptical fourth moment
\eqref{eq:ell-fourth-moment} \citep{Tyler1981,Browne1984,Muirhead1982},
\[
  \sigma^2_W=\sum_{ijkl}G_{ij}G_{kl}\,\lim_M M\,\mathrm{Cov}(S_{ij},S_{kl})
  =\sum_{ijkl}G_{ij}G_{kl}\bigl[(1+\kappa)(\Sigma_{ik}\Sigma_{jl}+\Sigma_{il}\Sigma_{jk})
  +\kappa\,\Sigma_{ij}\Sigma_{kl}\bigr].
\]
The three index contractions give
$\sum_{ijkl}G_{ij}G_{kl}\Sigma_{ik}\Sigma_{jl}
=\sum_{ijkl}G_{ij}G_{kl}\Sigma_{il}\Sigma_{jk}=\mathrm{tr}[(G\Sigma)^2]$
(using $G=G^\top$) and
$\sum_{ijkl}G_{ij}G_{kl}\Sigma_{ij}\Sigma_{kl}=\langle G,\Sigma\rangle^2$, so
\eqref{eq:si-variance} follows,
$\sigma^2_W=2(1+\kappa)\,\mathrm{tr}[(G\Sigma)^2]+\kappa\langle G,\Sigma\rangle^2$.

\emph{(ii) Immunity.}
If $W(b\Sigma)=W(\Sigma)$ for $b$ in a neighborhood of $1$, differentiating at $b=1$
gives $0=\frac{d}{db}W(b\Sigma)\big|_{b=1}=\langle G,\Sigma\rangle$ (Euler's relation
for homogeneity of degree zero).
The common-shock term therefore vanishes and
$\sigma^2_W(\Sigma,\kappa)=2(1+\kappa)\mathrm{tr}[(G\Sigma)^2]
=(1+\kappa)\sigma^2_W(\Sigma,0)$.

\emph{(iii) Converse.}
Fix $\Sigma\in\mathcal{C}$.
By (i), $\sigma^2_W(\Sigma,\kappa_0)-(1+\kappa_0)\sigma^2_W(\Sigma,0)
=\kappa_0\langle\nabla W(\Sigma),\Sigma\rangle^2$.
The hypothesis makes the left side zero, and $\kappa_0>0$ forces
$\langle\nabla W(\Sigma),\Sigma\rangle=0$ on $\mathcal{C}$.
Now fix $\Sigma$ and let $J=\{b>0:b\Sigma\in\mathcal{C}\}$, an interval by
assumption, and define $\phi(b)=W(b\Sigma)$ on $J$.
Then $\phi'(b)=\langle\nabla W(b\Sigma),\Sigma\rangle
=b^{-1}\langle\nabla W(b\Sigma),b\Sigma\rangle=0$ for every $b\in J$, so $\phi$
is constant on $J$, so $W$ is scale-invariant along every ray segment in
$\mathcal{C}$.

\subsection[Gradients and Euler identities]{Gradients and Euler identities (Lemma~\ref{lem:ar-orth}, Proposition~\ref{prop:ar-clt})}\label{app:ar-orth}

Let $\Sigma=\sum_i\lambda_iu_iu_i^\top$, $T=\mathrm{tr}\,\Sigma$,
$S_K=\sum_{i\le K}\lambda_i$, $p_i=\lambda_i/T$.

For $\AR_K$, write $\AR_K(\Sigma)=\mathrm{tr}(P_K\Sigma)/\mathrm{tr}\,\Sigma$
with $P_K$ the top-$K$ spectral projector, well-defined and differentiable when
$\lambda_K>\lambda_{K+1}$.
For a symmetric perturbation $d\Sigma$, the first-order projector change $dP_K$
has zero diagonal blocks in the eigenbasis (it maps the top block to its
complement), so $\mathrm{tr}(dP_K\,\Sigma)=0$ and
\[
  d\AR_K=\frac{\mathrm{tr}(P_K\,d\Sigma)}{T}
  -\frac{S_K\,\mathrm{tr}(d\Sigma)}{T^2}
  \;\Rightarrow\;
  \nabla\AR_K=\frac{P_K}{T}-\frac{S_K}{T^2}I.
\]
The Euler identity follows, $\langle\nabla\AR_K,\Sigma\rangle=S_K/T-S_KT/T^2=0$.
For the variance, $G\Sigma$ is diagonal in the eigenbasis with entries
$\lambda_i(1/T-S_K/T^2)$ for $i\le K$ and $-\lambda_iS_K/T^2$ for $i>K$, so
$\mathrm{tr}[(G\Sigma)^2]=T^{-4}\bigl[(T-S_K)^2\sum_{i\le K}\lambda_i^2
+S_K^2\sum_{i>K}\lambda_i^2\bigr]$, giving the displayed $\AR_K$ variance in
Proposition~\ref{prop:ar-clt}.

For $f_1$, the special case $K=1$ (with $\lambda_1$ simple) gives
$\nabla f_1=u_1u_1^\top/T-\lambda_1I/T^2$; the Euler identity
$\lambda_1/T-\lambda_1T/T^2=0$ holds, and the variance is
$\sigma^2_{f_1}=2(1+\kappa)T^{-2}[(1-f_1)^2\lambda_1^2+f_1^2\sum_{i>1}\lambda_i^2]$.

For the entropy $H=-\sum_ip_i\log p_i$, differentiating with
$\partial p_j/\partial\lambda_i=(\delta_{ij}-p_j)/T$ gives
\[
  \frac{\partial H}{\partial\lambda_i}
  =-\sum_j(\log p_j+1)\frac{\delta_{ij}-p_j}{T}
  =-\frac{\log p_i+1-\sum_jp_j(\log p_j+1)}{T}
  =-\frac{\log p_i+H}{T},
\]
so $\nabla H=-T^{-1}\sum_i(\log p_i+H)u_iu_i^\top$.
The Euler identity reads $\sum_i\lambda_i\,\partial_iH
=-\sum_ip_i(\log p_i+H)=-(-H+H)=0$, and the variance is
$\sigma^2_H=2(1+\kappa)T^{-2}\sum_i\lambda_i^2(\log p_i+H)^2$.

Finally, in eigenvalue coordinates the $\AR_K$ gradient is
$g_i=(\mathbf 1_{\{i\le K\}}T-S_K)/T^2$ with
$\sum_ig_i\lambda_i=(TS_K-S_KT)/T^2=0$, matching the statement.

\subsection[Exact window-scale pivotality]{Exact window-scale pivotality (Proposition~\ref{prop:exact-pivot})}\label{app:exact-pivot}

With $r_\tau=c\,z_\tau$ for all $\tau$ in the window and a single positive
random variable $c$, the demeaned returns scale identically, so
$S_r=\frac{1}{M-1}\sum_\tau\tilde r_\tau\tilde r_\tau^\top=c^2S_z$
\emph{realization by realization}.
Scale invariance then gives $W(S_r)=W(c^2S_z)=W(S_z)$ pathwise; in particular
the law of $W(S_r)$ coincides with that of $W(S_z)$ whatever the joint law of
$c$, and no independence between $c$ and $\{z_\tau\}$ is used.

For the monitor, suppose both windows share the scale ($r=cz$ on the union of
the two windows) and $\mathcal{A}$ is scale-equivariant. Then
$C_t=\mathcal{A}(c^2S_{z,t})=c^2\mathcal{A}(S_{z,t})$, so eigenvectors and hence
$\widehat P_{K,t}$, $\widehat D_{K,t}$ are unchanged, while
$\eta_t=\|C_t-S_t\|_{\op}$ and $\widehat\Delta_{K,t}$ both scale by $c^2$, so
every ratio $\eta_s/\widehat\Delta_{K,s}$, hence $T_{K,s}$, $\tau^*_{K,t}$, and
all flags, is invariant.
Scale-equivariance of linear shrinkage and QIS follows from their definitions, since
both estimate shrinkage targets and intensities from scale-free spectral
quantities applied multiplicatively to the data scale.

\subsection[The cleaning--debiasing wedge]{The cleaning--debiasing wedge (Proposition~\ref{prop:wedge})}\label{app:wedge}

The proof computes the almost-sure limits of the sample and oracle ratios from spiked-model
asymptotics, characterizes the trace-preserving maps that remove the wedge, and concludes
with the central limit theorem for the debiased estimator.
Throughout, $N/M\to c\in(0,1)$, the spikes are fixed, distinct, and above the
edge $\sigma^2(1+\sqrt c)$, and $u_i$ ($v_i$) denote sample (population)
eigenvectors.

\emph{(i) Sample wedge.}
By \citet{BaikSilverstein2006,Paul2007},
$\widehat\lambda_i\to_{a.s.}\psi(\lambda_i)
=\lambda_i+c\sigma^2\lambda_i/(\lambda_i-\sigma^2)$ for each spike
$i\le K'$, so the numerator of $\widehat{\AR}_K(S)$ satisfies
$\sum_{i\le K}\widehat\lambda_i\to S_K+c\sigma^2\sum_{i\le K}
\lambda_i/(\lambda_i-\sigma^2)$.
For the denominator, $\E\,\mathrm{tr}\,S=\mathrm{tr}\,\Sigma_N$ exactly and
$\mathrm{Var}(\mathrm{tr}\,S)=2\,\mathrm{tr}(\Sigma_N^2)/M=O(N/M)$, so
$\mathrm{tr}\,S/\mathrm{tr}\,\Sigma_N\to_{a.s.}1$ (fourth-moment
Borel--Cantelli along $N\asymp cM$).
Dividing and using $\AR_K(\Sigma_N)=S_K/\mathrm{tr}\,\Sigma_N$ gives the claim.

\emph{(ii) Oracle wedge.}
Decompose
$u_i^\top\Sigma u_i=\sigma^2+\sum_{j\le K'}(\lambda_j-\sigma^2)
\langle u_i,v_j\rangle^2$.
By \citet{Paul2007,BenaychGeorgesNadakuditi2011},
$\langle u_i,v_i\rangle^2\to_{a.s.}\alpha_i^2$ and
$\langle u_i,v_j\rangle^2\to_{a.s.}0$ for $j\ne i$, giving
$d_i^*\to\sigma^2+(\lambda_i-\sigma^2)\alpha_i^2$ for spike directions.
Trace preservation is exact,
$\sum_i d_i^*=\sum_i u_i^\top\Sigma u_i=\mathrm{tr}\,\Sigma$.
Order preservation for separated spikes (spike limits exceed
$\sigma^2$, bulk values tend to $\sigma^2+o(1)$ by delocalization) puts the
top-$K$ oracle values on the spike directions, so
\[
  \widehat{\AR}_K(D^*)\to
  \frac{S_K-\sum_{i\le K}(\lambda_i-\sigma^2)(1-\alpha_i^2)}{\mathrm{tr}\,\Sigma_N}.
\]
The stated identity follows from algebra. With $\rho=(\lambda-\sigma^2)/\sigma^2$,
\[
  1-\alpha^2
  =\frac{c/\rho^2+c/\rho}{1+c/\rho}
  =\frac{c(1+\rho)}{\rho(\rho+c)}
  \;\Rightarrow\;
  (\lambda-\sigma^2)(1-\alpha^2)
  =\sigma^2\rho\cdot\frac{c(1+\rho)}{\rho(\rho+c)}
  =\frac{c\sigma^2\lambda}{\lambda-\sigma^2+c\sigma^2},
\]
using $\sigma^2(1+\rho)=\lambda$.
Summing over $i\le K$ gives $\Delta_O$, and $\Delta_O<\Delta_S$ holds term by
term since the denominators differ by $+c\sigma^2>0$.
For strong spikes ($\lambda_i\gg\sigma^2$) both wedges approach
$Kc\sigma^2/S_K$, so the oracle under-counts systemic concentration by
asymptotically the same amount that the raw sample over-counts it.

\emph{Zero-wedge characterization.}
The claim referenced in \S\ref{sec:ar-hd-caution} is that within
trace-preserving rotation-equivariant estimators whose spike eigenvalues
converge a.s.\ to a deterministic map $m(\lambda_i)$, the relative $\AR_K$
wedge vanishes for all spike configurations in an open set if and only if $m$
is the identity above the bulk edge.
For a trace-preserving estimator with deterministic spike map $m$, the same argument
gives $\widehat{\AR}_K\to\sum_{i\le K}m(\lambda_i)/\mathrm{tr}\,\Sigma_N$
(bulk values cannot exceed the edge image, preserving order for separated
spikes), so the relative wedge vanishes iff
$\sum_{i\le K}m(\lambda_i)=\sum_{i\le K}\lambda_i$.
Varying a single spike $\lambda_1$ over an open interval while holding the
others fixed forces $m(\lambda)=\lambda+\mathrm{const}$, and letting two
configurations share $K-1$ spikes forces the constant to zero; hence $m=$ id
above the edge.
The Frobenius-optimal map satisfies
$m^*(\lambda)-\lambda=-(\lambda-\sigma^2)(1-\alpha^2)<0$ for $c>0$.

\emph{(iii) Debiased CLT.}
Since $\psi$ is strictly increasing above the edge with
$\psi'(\lambda)=1-c\sigma^4/(\lambda-\sigma^2)^2>0$,
$\psi^{-1}$ is well-defined and smooth; consistency follows from (i) and the
continuous mapping theorem.
By \citet{BaiYao2008}, $\sqrt M(\widehat\lambda_i-\psi(\lambda_i))$ are
asymptotically independent Gaussians across separated spikes, so the delta
method through $\psi^{-1}$ makes the numerator
$\sum_{i\le K}\psi^{-1}(\widehat\lambda_i)$ asymptotically normal around $S_K$
at rate $\sqrt M$.
The denominator fluctuation is
$\sqrt M\,(\mathrm{tr}\,S-\mathrm{tr}\,\Sigma_N)/\mathrm{tr}\,\Sigma_N
=O_p(\sqrt{M}\cdot\sqrt{N/M}/N)=O_p(N^{-1/2})\to0$, negligible at the $\sqrt M$
scale.
Since $\AR_K(\Sigma_N)=S_K/\mathrm{tr}\,\Sigma_N=O(1/N)$, the absolute statistic
$\sqrt M(\widehat{\AR}{}^{\mathrm{deb}}_K-\AR_K)=O_p(1/N)\to0$ is degenerate;
the non-degenerate limit is for the relative error,
$\sqrt M(\widehat{\AR}{}^{\mathrm{deb}}_K/\AR_K-1)\to_d N(0,V_{\mathrm{rel}})$
with $V_{\mathrm{rel}}=\mathrm{Var}_\infty[\sqrt M(\sum_{i\le K}\psi^{-1}(\widehat\lambda_i)-S_K)]/S_K^2$
of order one.
The debiased estimator is trace-preserving. The spike inflation
$\sum_i\bigl(\widehat\lambda_i-\psi^{-1}(\widehat\lambda_i)\bigr)$ is returned
to the bulk eigenvalues uniformly, so the debiased spectrum sums to
$\mathrm{tr}\,S$ and the plug-in ratio equals
$\sum_{i\le K}\psi^{-1}(\widehat\lambda_i)/\mathrm{tr}\,S$.
Without the redistribution the denominator would be depleted by
\[
  c\sigma^2\sum_{i\le K'}\frac{\lambda_i}{\lambda_i-\sigma^2}+o_{a.s.}(1),
\]
a relative upward $\AR_K$ error equal, to leading order, to $\AR_{K'}(\Sigma_N)$
times the level-$K'$ sample wedge; it vanishes asymptotically at rate $O(c/N)$
but amounts to roughly $40\%$ of the wedge itself at the panel concentrations of
\S\ref{sec:mc-si-coverage}.

\newpage
\section{Equity panel illustration}
\label{app:panel}

The panel comprises daily adjusted returns on 115 U.S.\ large-cap stocks (retrieved
from Yahoo Finance via the \texttt{yfinance} interface), with covariance estimates
formed on rolling windows of length $M{=}252$ and shrinkage by the QIS estimator.
This appendix reports what the procedures of the main text produce when the
population covariance $\Sigma_t$ is unknown, namely projector movement against
the worst-case band, the calibrated tests by subperiod, and a calm-period
calibration of the analytic threshold.
All confirmatory evidence in the paper is simulation under known $\Sigma$
(\S\ref{sec:mc}); on the panel the perturbation size is the sample-relative
quantity $\eta_t=\|C_t-S_t\|_{\op}$, and flags are diagnostics rather than tests
of latent structural change.

\begin{remark}[Scope of the panel evidence]
\label{rem:empirical-scope}
On the panel, $\|C_t-S_t\|_{\op}$ measures deviation from the window sample $S_t$,
not from $\Sigma_t$, so $\widehat D_{K,t}>\alpha\tau^*_{K,t}$ is a diagnostic flag,
not evidence that the latent $D_{K,t}\neq 0$ (\S\ref{sec:theory-monitor}).
\end{remark}

\subsection{The worst-case band on the panel}
\label{sec:empirical}

Table~\ref{tab:monitor} reports medians of the observed movement
$\widehat D_{K,t}$, the capped empirical-gap band $\widehat\tau^*_{K,t}$
(Corollary~\ref{cor:empirical-gap}), and the weak-identification and rank-change
shares by subperiod.
The band behaves as the theory predicts.
In calm periods it sits far above the observed movement and never flags; at the
calibrated scale $\alpha^\star{=}0.75$ of \S\ref{sec:mc-alpha} the flag rate is
$0.3\%$ over the full sample.
In March--April 2020 the normalized eigengap shrinks to near zero (median
$\gaptwelve\approx 0.02$) while $\widehat D_{K,t}$ stays moderate; the per-date
Davis--Kahan terms saturate and $\widehat\tau^*_{K,t}$ rests at its ceiling
$\sqrt{2\widehat K}$ because $\widehat\Delta_{K,t}\le 2\eta_t$ on $82.6\%$ of
dates (weak identification, Corollary~\ref{cor:empirical-gap}).
On real data the band behaves as intended. When identification is weak it
carries no information beyond the trivial bound, and the calibrated tests below
take over.

\begin{table}[!ht]
  \centering
  \caption{Gap-adjusted monitoring on the panel ($\eta_t=\|C_t-S_t\|_{\op}$).}
  \label{tab:monitor}
  \footnotesize
  \begin{tabular}{lrrrrrr}
    \toprule
    Period & $n$ & Med.\ $\widehat D$ & Med.\ norm.\ gap & Med.\ $\widehat\tau^*$ &
    \% weak id. & \% rank chg. \\
    \midrule
    Calm 17--19 & 754 & 0.0082 & 0.80 & 0.45 & 0.0 & 0.0 \\
    COVID-20 & 92 & 0.035 & 0.02 & 2.45 & 82.6 & 5.4 \\
    Full & 6024 & 0.0095 & 0.81 & 0.44 & 7.7 & 0.5 \\
    \bottomrule
  \end{tabular}
  \par\vspace{0.4em}
  \parbox{\linewidth}{\footnotesize $\widehat\tau^*_{K,t}$ is the capped empirical-gap band
  (Cor.~\ref{cor:empirical-gap}); during COVID it saturates at its ceiling
  $\sqrt{2\widehat K}$ ($\widehat K{=}3$, $\sqrt{6}\approx 2.45$) because
  $\widehat\Delta_{K,t}\le 2\eta_t$ on 82.6\% of dates, reflecting weak identification rather than large
  observed $\widehat D_{K,t}$. Rank-change dates ($\widehat K_t\ne\widehat K_{t-1}$,
  \S\ref{sec:theory-monitor}) are reported separately.}
\end{table}

\subsection{Panel diagnostics by subperiod}
\label{sec:panel-inference}

Table~\ref{tab:panel-inference} applies the analytic and bootstrap thresholds to
the panel ($M{=}252$, daily stride $s{=}1$, QIS monitor; bootstrap on every 21st
date, $B{=}60$; elliptical-scaled nulls of Remark~\ref{rem:elliptical-null} with
per-window moment-based $\widehat\kappa_t$, median $\widehat\kappa\approx 1.5$).
Real returns are neither i.i.d.\ nor exactly elliptical, so the per-date flag
rates are reported as diagnostics rather than as size-controlled tests.
The moment-based $\widehat\kappa$ (median $\approx1.5$) and the radial MLE
($\approx0.2$) disagree for the same reason, and the moment scaling, which absorbs
within-window volatility clustering, is the better-calibrated choice on the panel
(calm flag rate $2.7\%$ versus $7.6\%$).
Calm 2017--19 flags at $2.7\%$ and the full sample at $4.1\%$, near the
nominal $5\%$ that an i.i.d.-elliptical within-window null would target, while COVID-2020
flags at $21.7\%$/$25.0\%$ (analytic/bootstrap) and the GFC at
$11.5\%$/$18.2\%$, a regime gradient the worst-case band cannot express
(Figure~\ref{fig:panel-dhat}).
Median AR$_{\widehat K}$ rises from $0.31$ (calm) to $0.71$ (COVID) with
delta-method CI half-widths of about $0.06$; absorption-ratio movements smaller
than this are not distinguishable from estimation noise at $M{=}252$, a caveat
directly relevant to systemic-risk monitors built on \citet{Kritzman2011}.

\begin{table}[!ht]
  \centering
  \caption{Panel inference by subperiod (QIS, $M{=}252$, $s{=}1$), with observed movement,
  analytic and bootstrap thresholds and flags, absorption ratio with elliptical delta CI.}
  \label{tab:panel-inference}
  \footnotesize
  \begin{tabular}{lrrrrrrr}
\toprule
Period & $n$ & Med.\ $\widehat D$ & Med.\ $q^{\mathrm{an,ell}}_{95}$ & Flags anl.\ (\%) & Flags boot.\ (\%) & Med.\ $\AR_{\widehat K}$ & Med.\ CI width \\
\midrule
Full & 6024 & 0.0095 & 0.0361 & 4.1\% & 4.1\% & 0.327 & 0.121 \\
Calm 17--19 & 754 & 0.0082 & 0.0362 & 2.7\% & 0.0\% & 0.311 & 0.126 \\
GFC 08--09 & 209 & 0.0093 & 0.0251 & 11.5\% & 18.2\% & 0.494 & 0.136 \\
COVID-20 & 92 & 0.0351 & 0.1345 & 21.7\% & 25.0\% & 0.713 & 0.132 \\
Hikes 22--23 & 462 & 0.0071 & 0.0253 & 2.4\% & 9.1\% & 0.407 & 0.122 \\
\bottomrule
\end{tabular}

\end{table}

\begin{figure}[!ht]
  \centering
  \includegraphics[width=\linewidth]{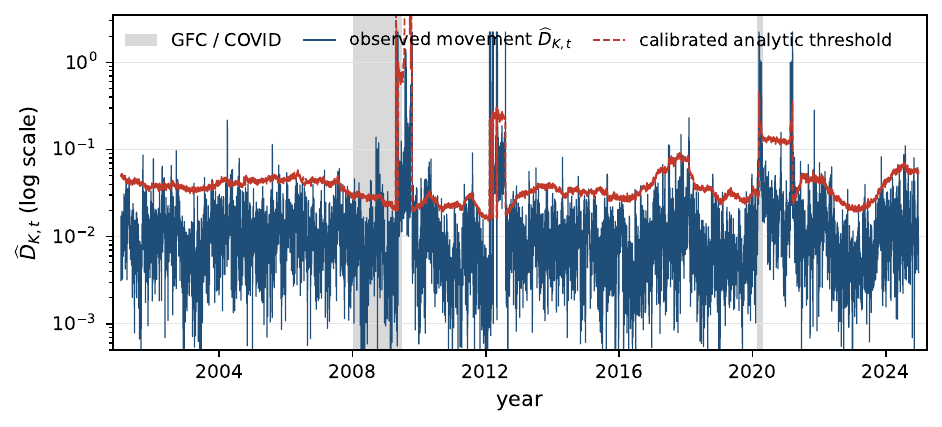}
  \caption{Daily projector movement $\widehat D_{K,t}$ on the panel against the
  elliptical-scaled analytic threshold (moment-based $\widehat\kappa_t$,
  \S\ref{sec:ar-inference}), with the 2008 financial crisis and the 2020 COVID
  window shaded. Observed movement stays below the threshold in calm periods and
  exceeds it in clusters during market stress; the threshold widens, and can
  saturate, when identification is weak, so a few weak-identification dates fall
  above the plotted range.}
  \label{fig:panel-dhat}
\end{figure}

Table~\ref{tab:event-study} conditions forward outcomes on the analytic-test
flag.
Over the 22 trading days following a flagged date, equal-weight portfolio
volatility averages $26.2\%$ annualized versus $15.9\%$ after unflagged dates
(ratio $1.65$; moving-block bootstrap $95\%$ CI $[1.22,2.15]$), and the mean
forward maximum drawdown is $4.8\%$ versus $2.8\%$.
Within the most stressed trailing-volatility quintile the forward-volatility
ratio is $1.70$, so the flags carry information beyond the current volatility
level; in mid-volatility regimes the flag adds little, consistent with the
detectability frontier of \S\ref{sec:tier-power}.

\begin{table}[!ht]
  \centering
  \caption{Event study of forward 22-day outcomes after analytic-test flag vs.\ no-flag dates,
  stratified by trailing 22-day volatility quintile (EW portfolio; annualized vol;
  DD = maximum drawdown).}
  \label{tab:event-study}
  \footnotesize
  \begin{tabular}{lrrrrrr}
\toprule
Trailing-vol quintile & $n_{\mathrm{flag}}$ & $n_{\mathrm{no}}$ & Fwd vol (flag) & Fwd vol (no) & Ratio & Fwd DD (flag/no) \\
\midrule
Q1 & 27 & 1174 & 11.8\% & 10.7\% & 1.10 & 1.9\%/1.8\% \\
Q2 & 31 & 1169 & 14.4\% & 12.0\% & 1.20 & 3.0\%/2.1\% \\
Q3 & 38 & 1163 & 14.1\% & 14.4\% & 0.98 & 2.5\%/2.8\% \\
Q4 & 61 & 1139 & 18.0\% & 16.3\% & 1.10 & 3.0\%/2.7\% \\
Q5 & 90 & 1111 & 45.2\% & 26.5\% & 1.70 & 8.5\%/4.7\% \\
\midrule
Pooled & 247 & 5756 & 26.2\% & 15.9\% & 1.65 & 4.8\%/2.8\% \\
\bottomrule
\end{tabular}

\end{table}

\subsection{Calm-period calibration of the analytic test}
\label{app:monitor-calib}

As a complementary check, the threshold $\tau^{\mathrm{cal}}$ is set directly
from calm 2017--19 history (95th percentile or block-bootstrap upper quantile of
$\widehat D_{K,t}$, block length 22).
Table~\ref{tab:monitor-calib-inference} reports the resulting flag prevalence by
subperiod against the uncalibrated worst-case band; the calm-calibrated
threshold reproduces the regime gradient of Table~\ref{tab:panel-inference}.

\begin{table}[!ht]
  \centering
  \caption{Panel monitoring flag prevalence, worst-case band (uncalibrated $\tau^*_{K,t}$) vs calm-calibrated analytic threshold ($\tau^{\mathrm{cal}}$).}
  \label{tab:monitor-calib-inference}
  \footnotesize
  \begin{tabular}{lrrr}
\toprule
Period & $n$ & Band (\%) & Analytic quantile (\%) \\
\midrule
Full & 6024 & 0.0 & 7.3 \\
Calm 2017--19 & 754 & 0.0 & 5.0 \\
GFC 2008--09 & 209 & 0.0 & 14.4 \\
COVID 2020 & 92 & 0.0 & 46.7 \\
Hikes 2022--23 & 462 & 0.0 & 2.2 \\
High-vol tertile & 1922 & 0.0 & 10.7 \\
\bottomrule
\end{tabular}

\end{table}

\bibliographystyle{plainnat}
\bibliography{references}

\end{document}